\documentclass[a4paper,UKenglish,cleveref, autoref, thm-restate]{lipics-v2021}

\hideLIPIcs  

\usepackage[subtle]{savetrees} 
 
\usepackage{paralist} 
\usepackage{amsthm,amsmath,amssymb} 
\usepackage{xspace}
\usepackage[vlined,linesnumbered,ruled]{algorithm2e}
\numberwithin{equation}{section}
\usepackage{csquotes}
\usepackage{pgfplots}
\pgfplotsset{compat=newest}
\usepgfplotslibrary{groupplots}

\usepackage{cite} 

\usepackage{tikz}
\usepackage{graphics,graphicx,xcolor,color,url}    
\usetikzlibrary{arrows,shapes,snakes,automata,backgrounds,petri}
\tikzset{
	cross/.style={cross out, draw=black, minimum size=2*(#1-\pgflinewidth), inner sep=0pt, outer sep=0pt},
	cross/.default={3pt},
	vertex/.style={circle, draw, fill=black, inner sep=0pt, minimum width=4pt},
}

\newcommand{\drawpred}[2]{
	\draw (#1, #2) node[cross,red] {};
}
\newcommand{\drawreal}[2]{
	\draw[black!40!green] (#1, #2) circle (3pt);
}
\newcommand{\interval}[4]{
	\draw (#2, #4) node[anchor=east]{#1} -- (#3, #4);
	\draw (#2, #4-0.1) -- (#2, #4+0.1);
	\draw (#3, #4-0.1) -- (#3, #4+0.1);
}
\newcommand{\intervalp}[5]{
	\interval{#1}{#2}{#3}{#4}
	\drawpred{#5}{#4}
}

\newcommand{\intervalpr}[6]{
	\interval{#1}{#2}{#3}{#4}
	\drawpred{#5}{#4}
	\drawreal{#6}{#4}
}

\SetKw{KwLet}{let}
\SetKw{KwAnd}{and}

\newcommand{\bintervalpr}[6]{
	\binterval{#1}{#2}{#3}{#4}
	\drawpred{#5}{#4}
	\drawreal{#6}{#4}
}

\newcommand{\binterval}[4]{
	\draw[blue] (#2, #4) node[anchor=east,black]{#1} -- (#3, #4);
	\draw[blue] (#2, #4-0.1) -- (#2, #4+0.1);
	\draw[blue] (#3, #4-0.1) -- (#3, #4+0.1);
}

\title{Learning-Augmented Query Policies for Minimum Spanning Tree with Uncertainty} 

\titlerunning{Learning-Augmented Query Policies for MST with Uncertainty} 

\author{Thomas Erlebach}{Department of Computer Science, Durham University, United Kingdom \and \url{https://www.cs.le.ac.uk/people/te17/} }{te17@leicester.ac.uk}{https://orcid.org/0000-0002-4470-5868}{Supported by EPSRC grant EP/S033483/2.}
\author{Murilo Santos de Lima}{\camerate{Munich, Germany} \and \url{https://www.ime.usp.br/~mslima/} }{mslima@ic.unicamp.br}{https://orcid.org/0000-0002-2297-811X}{\camerate{Work done while employed at University of Leicester, United Kingdom, and} funded by EPSRC grant EP/S033483/1.}
\author{Nicole Megow}{Faculty of Mathematics and Computer Science, University of Bremen, Germany \and \url{https://www.uni-bremen.de/en/cslog/nmegow} }{nicole.megow@uni-bremen.de}{https://orcid.org/0000-0002-3531-7644}{Supported by the German Science Foundation (DFG) under contract ME~3825/1.}
\author{Jens Schlöter}{Faculty of Mathematics and Computer Science, University of Bremen, Germany \and \url{https://www.uni-bremen.de/en/cslog/team/jens-schloeter} }{jschloet@uni-bremen.de}{https://orcid.org/0000-0003-0555-4806}{Funded by the German Science Foundation (DFG) under contract ME~3825/1.}

\authorrunning{T. Erlebach, M.S. de Lima, N. Megow, J. Schlöter} 

\Copyright{Thomas Erlebach, Murilo Santos de Lima, Nicole Megow, Jens Schlöter} 



\category{} 

\relatedversion{} 

\nolinenumbers 

\newcommand{\pred}[1]{\overline{#1}}
\newcommand{\vertDOTS}{\begin{array}{lll}
		\bullet\\
		\bullet\\
		\bullet\\
	\end{array}
}

\newcommand{\jo}{\ensuremath{k^+}\xspace} 
\newcommand{\oj}{\ensuremath{k^-}\xspace} 

\newcommand{\ALG}{\mathrm{ALG}}
\newcommand{\OPT}{\mathrm{OPT}}
\newcommand{\opt}{|\OPT|}

\newcommand{\w}{\ensuremath{\overline{w}}} 

\newcommand{\ZZ}{\mathbb{Z}}

\newcommand{\RR}{\mathbb{R}}

\newcommand{\eps}{\ensuremath{\varepsilon}\xspace} 
\newcommand{\sym}{\ensuremath{\Delta}\xspace} 
\DeclareMathOperator{\EX}{\mathbb{E}}
\newcommand{\ud}{\ensuremath{D}}
\newcommand{\hs}{\ensuremath{\mathcal{H}}}
\usepackage[textsize=scriptsize]{todonotes}

\newcommand{\camera}[1]{{#1}}
\newcommand{\camerate}[1]{{#1}}

\newcommand{\red}[1]{#1}
\newcommand{\nnew}[1]{#1}
\newcommand{\tomc}[1]{#1}
\newcommand{\jnew}[1]{#1}

\graphicspath{ {figs/} }

\newtheorem{obs}[theorem]{Observation}

\EventEditors{John Q. Open and Joan R. Access}
\EventNoEds{2}
\EventLongTitle{42nd Conference on Very Important Topics (CVIT 2016)}
\EventShortTitle{CVIT 2016}
\EventAcronym{CVIT}
\EventYear{2016}
\EventDate{December 24--27, 2016}
\EventLocation{Little Whinging, United Kingdom}
\EventLogo{}
\SeriesVolume{42}
\ArticleNo{23}		

\begin{document}

\maketitle

\begin{abstract}
	We study how to utilize (possibly erroneous) predictions in a model for computing under uncertainty in which an algorithm can query unknown data. Our aim is to minimize the number of queries needed to solve the minimum spanning tree problem, a fundamental combinatorial optimization problem that has been central also to the research area of explorable uncertainty. For all integral $\gamma\ge 2$, we present algorithms that are $\gamma$-robust and $(1+\frac{1}{\gamma})$-consistent, meaning that they use at most $\gamma\OPT$ queries if the predictions are arbitrarily wrong and at most $(1+\frac{1}{\gamma})\OPT$ queries if the predictions are correct, where $\OPT$ is the optimal number of queries for the given instance. Moreover, we show that this trade-off is best possible. Furthermore, we argue that a suitably defined \emph{hop distance} is a useful measure for the amount of prediction error and design algorithms with performance guarantees that degrade smoothly with the hop distance. We also show that the predictions are PAC-learnable in our model. Our results demonstrate that untrusted predictions can circumvent the known lower bound of~$2$, without any degradation of the worst-case ratio.  To obtain our results, we provide new structural insights for the minimum spanning tree problem that might be useful in the context of 
		\nnew{query-based algorithms} regardless of predictions. \nnew{In particular, we generalize the concept of witness sets---the key to lower-bounding the optimum---by proposing novel global witness set structures and completely new ways of adaptively using those.}

\end{abstract} 

%

\section{Introduction}
We introduce learning-augmented algorithms to the area of optimization under explorable uncertainty and focus on the fundamental {\em minimum spanning tree (MST) problem under {explorable} uncertainty}. We are given a (multi)graph $G=(V,E)$ with unknown edge weights $w_e\in \RR_+$, for $e\in E$.  For each edge $e$, an uncertainty interval $I_e$ is known that contains $w_e$. A \emph{query} on edge $e$ 
reveals the true value $w_e$. The task is to determine an MST, i.e., a tree that connects all vertices of $G$, of minimum total weight w.r.t.\ the true values~$w_e$. A query set is called {\em feasible} if it reveals sufficient information to identify an MST (not necessarily its exact weight). As queries are costly, the goal is to find a feasible query set of minimum size. 

We study \emph{adaptive strategies} that make queries sequentially and utilize the results of previous steps to decide upon the next query. 
\tomc{As there} exist input instances that are impossible to solve without querying all \tomc{edges, we} evaluate our algorithms in an {\em instance-dependent} manner: For each input, we compare the number of queries made by an algorithm with the best possible number of queries \emph{for that input}, using competitive analysis. 
For a given problem instance, let $\OPT$ denote an arbitrary optimal query set \camera{(we later give a formal definition of $\OPT$)}. An algorithm is $\rho$-{\em competitive} if it executes, for any problem instance, at most $\rho \cdot {\opt}$ queries.
\camera{While MST under explorable uncertainty is not a classical online problem where the input is revealed passively over time, the query results are uncertain and, to a large degree, dictate whether decisions to query certain edges were good or not. 
\camerate{For analyzing an algorithm, it is natural to assume that the query results are determined by
an adversary.}
This gives the problem a clear online flavor and prohibits the existence of $1$-competitive algorithms even if we have unlimited running time and space~\cite{erlebach08steiner_uncertainty}. 
We note that competitive algorithms in general do not have any running time requirements, but all our algorithm run in polynomial time.}



The MST problem is among the most widely studied problems in the research area of {\em explorable uncertainty}~\cite{kahan91queries} and has been a cornerstone in the development of algorithmic approaches and lower bound techniques \cite{erlebach08steiner_uncertainty,erlebach14mstverification,erlebach15querysurvey,megow17mst,focke17mstexp,MerinoS19}. The best known deterministic algorithm for MST with uncertainty is $2$-competitive, and no deterministic algorithm can be better \cite{erlebach08steiner_uncertainty}. A randomized algorithm with competitive ratio $1.707$ is known~\cite{megow17mst}.
\jnew{Further work considers the non-adaptive problem, which has a very different flavor~\cite{MerinoS19}.
}


In this paper, we assume that an algorithm has, for each edge $e$, access to a prediction $\pred{w}_e \in I_e$ for the unknown value $w_e$. 
For example, machine learning~(ML) methods could be used to predict the value of an edge. 
{Given the tremendous progress in artificial intelligence and ML in recent decades, we can expect that those predictions are of good accuracy, but there is no guarantee and the predictions might be completely wrong.}
This lack of provable performance guarantees for ML
often causes concerns regarding how confident one can be that
an ML algorithm will perform sufficiently well in all circumstances.
We address the very natural question
whether the availability of such (ML) predictions can be exploited by query 
algorithms for computing with explorable uncertainty.
Ideally, an algorithm should perform very well if
predictions are accurate, but even if they are arbitrarily
wrong, the algorithm should not perform worse than an algorithm
without access to predictions.
To emphasize
that the predictions can be wrong, we refer to them
as \emph{untrusted predictions}. 

We note that the availability of \tomc{both uncertainty intervals and
untrusted predictions is} natural in many scenarios. For example,
the quality of links (measured using metrics such as throughput
and reliability) in a wireless network often fluctuates over
time within a certain interval, and 
ML methods can be used to predict the precise link quality based on time-series data
of previous link quality measurements~\cite{Abdel-NasserMOLP/20}.
The actual quality of a link can be obtained via a new measurement.
If we wish to build a minimum spanning tree using links that
currently have the highest link quality and want to minimize
the additional measurements needed, we arrive at an MST
problem with uncertainty and untrusted predictions.


We study for the first time the combination of
explorable uncertainty and untrusted predictions. Our work is inspired by the vibrant recent research trend of considering~untrusted (ML) predictions in the context
of {\em online} algorithms, a different uncertainty model where the input is revealed to an algorithm incrementally. Initial work on online caching problems~\cite{lykouris2018competitive} has initiated a vast growing line of research on  caching~\cite{rohatgi2019near,AntoniadisCEPS2020,Wei20}, rent-or-buy problems \cite{purohit2018improving,GollapudiP19,WeiZ20}, scheduling~\cite{purohit2018improving,angelopoulos2019online,LattanziLMV20,Mitzenmacher20,AzarLT21},  graph problems~\cite{KumarPSSV19,lindermayrMS22,EberleLMNS22} 	and many more.

We adopt the following notions 
introduced in~{\cite{lykouris2018competitive,purohit2018improving}}: 
An algorithm is \emph{$\alpha$-consistent} if it
is $\alpha$-competitive when the predictions are correct, and it is
\emph{$\beta$-robust} if it is $\beta$-competitive no matter how wrong the
predictions are. Furthermore, we are interested in a smooth transition
between the case with correct predictions and the case with arbitrarily
incorrect predictions. We aim for performance guarantees
that degrade gracefully with 
increasing prediction error.

Given predicted values for the uncertainty intervals, it is tempting to simply run an optimal 
algorithm under the assumption that the predictions are correct. 
This is obviously optimal with respect to consistency, but might give arbitrarily bad solutions in the case when the predictions are faulty. 
Instead of blindly 
trusting the predictions, we need more sophisticated strategies to be robust against prediction errors. This requires new lower bounds on an optimal solution, new structural \tomc{insights,} and new algorithmic techniques.


\paragraph*{Main results} 
In this work, we show that, in the setting of explorable uncertainty,
it is in fact possible to exploit ML predictions of the uncertain
values and improve the performance of a query strategy when the predictions are good, while at the same time guaranteeing a strong bound on the worst-case performance even when the predictions
are arbitrarily bad.


We give algorithms for the MST problem with uncertainty that are parameterized by a hyperparameter $\gamma$ that reflects the user's confidence in the accuracy of the predictor. 
For any integral $\gamma\ge 2$, we present \tomc{a} $(1+\frac{1}{\gamma})$-consistent and $\gamma$-robust algorithm,
and show that this is the best possible trade-off between consistency and robustness. 
In particular, for $\gamma=2$, we obtain a $2$-robust, $1.5$-consistent algorithm.
 {It is worth noting that this algorithm achieves the
improved competitive ratio of $1.5$ 
for
accurate predictions
while maintaining the worst-case ratio of~$2$. This is in contrast
to 
many learning-augmented online algorithms where the exploitation
of predictions usually incurs an increase in the worst-case
ratio (e.g., \cite{purohit2018improving,AntoniadisGKK20}).} 

Our main result is a \camera{second and different} algorithm with a more fine-grained performance analysis that obtains a guarantee that improves with the accuracy of the predictions.
Very natural, simple error measures such as the number of inaccurate predictions or the $\ell_1$-norm of the difference between predictions and true values turn out to prohibit any reasonable error-dependency.
Therefore, we propose an error measure, called \emph{hop distance} $k_h$, that takes structural insights about uncertainty intervals into account
and may also be useful for other problems in computing with uncertainty and untrusted predictions.
We give a precise definition of this error measure \tomc{later. We} also show in~\Cref{app:learning} that the predictions are efficiently PAC-learnable with respect to~$k_h$. 
Our main result is a learning-augmented algorithm with a competitive ratio with a linear error-dependency $\min\{(1+ \frac{1}{\gamma}) + \frac{5 \cdot k_h}{|\OPT|}, \gamma+1 \}$, for any integral $\gamma \geq 2$. 
\jnew{All our algorithms have polynomial running-times.}
\nnew{We describe our techniques and highlight their novelty in the following section.}

{The integrality requirement for $\gamma$ comes from using $\gamma$ to determine set sizes and can be removed by randomization at the cost of a slightly worse consistency guarantee; see \Cref{app:rational-gamma}. 

\paragraph*{Further related work}

There is a long history of research on the tradeoff between exploration and exploitation when coping with uncertainty in the input data. Often, stochastic models are assumed, e.g., in work on multi-armed bandits~\cite{Thompson33,BubeckC12,GittinsGW11-book}, Weitzman's Pandora's box 
\cite{Weitzman1979}, and query-variants of combinatorial optimization problems; {see, e.g.,   
\cite{singla2018price,gupta2019markovian,yamaguchi18ipqueries} and many more. 
In our work, we assume no knowledge of stochastic information and aim for robust algorithms that perform well even in a worst case. 

The line of research on 
explorable uncertainty has been initiated by Kahan~\cite{kahan91queries} in the context of selection problems.
Subsequent work addressed finding the~$k$-th smallest value in a set of uncertainty intervals~\cite{gupta16queryselection,feder03medianqueries}, 
caching problems 
\cite{olston2000queries}, computing a function value~\cite{khanna01queries}, sorting~\cite{HalldorssonL21-sorting}, and classical combinatorial optimization problems. Some of the major aforementioned results on the MST problem under explorable uncertainty have been extended to general matroids~\cite{erlebach16cheapestset,megow17mst,MerinoS19}. Further problems that have been studied are the shortest path problem~\cite{feder07pathsqueires}, the knapsack problem~\cite{goerigk15knapsackqueries} and scheduling problems~\cite{DurrEMM20,arantes18schedulingqueries,albersE2020,AlbersE2021,BampisDKLP21}. 
Although optimization under explorable uncertainty has been studied mainly in an adversarial model, recently first results have been obtained for stochastic variants for sorting~\cite{ChaplickHLT21} and selection type problems (hypergraph orientation)~\cite{BampisDEdLMS21}.


There is a significant body of work on computing in models where information
about a hidden object can be accessed only via queries. The 
hidden object can, for example, be a function, a matrix, or a graph. 
In the graph context, property testing~\cite{Goldreich2017} has been studied extensively and there are many more works, see~\cite{Mazzawi2010,Beame2018,Rubinstein2018,Chen2020,AssadiCK21,Nisan2021}. The bounds on the number of queries made by an algorithm are typically  absolute (as a function of the input) and the resulting correctness guarantees are probabilistic. This is very different from our work, where we aim for a comparison to the minimum number of queries needed for the given graph.


\section{Overview of techniques and definition of error measure}
\label{sec:overview}

We assume that each uncertainty interval is either {\em open}, $I_e = (L_e,U_e)$,  or {\em trivial}, $I_e = \{w_e\}$, \nnew{and we refer to edge $e$ as {\em non-trivial} or {\em trivial}, respectively}; a standard 
assumption to avoid a simple lower bound of~$|E|$ on the competitive ratio \cite{gupta16queryselection,erlebach08steiner_uncertainty}.

\camera{Before we give an overview of the used techniques, we formally define feasible and optimal query sets.
We say that a query set $Q \subseteq E$ is \emph{feasible} if there exists a set of edges $T \subseteq E$ such that $T$ is an MST for the true values $w_e$ of all $e \in Q$ and \emph{every possible} combination of edge weights in $I_e$ for the unqueried edges $e \in E\setminus Q$.
That is, querying a feasible query set $Q$ must give us sufficient information to identify a spanning tree $T$ that is an MST for the true values no matter what the true values of the unqueried edges $E\setminus Q$ actually are.
We call a feasible query set $Q$ \emph{optimal} if it has minimum cardinality $|Q|$ among all feasible query sets.
Thus, the optimal solution depends only on the true values and not on the predicted values. 
}

\camera{As Erlebach \camerate{and Hoffmann\cite{erlebach14mstverification}} give a polynomial-time algorithm that computes an optimal query set under the assumption that all query results are known upfront, we can use their algorithm to compute the optimal query set under the assumption that all predicted values match the actual edge weights and query the computed set in an arbitrary order. If the predicted values are indeed correct, this yields a $1$-consistent algorithm.}
However, such an algorithm may have an arbitrarily bad performance if the predictions are incorrect. Similarly, the known deterministic $2$-competitive algorithm for the MST problem with uncertainty (without predictions)~\cite{erlebach08steiner_uncertainty} is $2$-robust and $2$-consistent. The known lower bound of $2$ rules out any robustness factor less than $2$. 
{It builds on the following simple example with two intersecting intervals $I_a,I_b$ that are candidates for the largest edge weight in a cycle. No matter which interval a deterministic algorithm queries first, say $I_a$, the realized value could be $w_a\in I_a\cap I_b$, which requires a second query. If the adversary chooses $w_b\notin I_a\cap I_b$, querying just $I_b$ would have been sufficient to identify the interval with larger true value.} 
\camera{See also~\cite[Example 3.8]{erlebach08steiner_uncertainty} and~\cite[Section 3]{megow17mst} for an illustration of the lower bound example.}

\medskip \noindent {\bf Algorithm overview}\quad  
We aim for $(1+\frac{1}{\gamma})$-consistent and $\gamma$-robust algorithms for each integral $\gamma \ge 2$. 
The algorithm
proceeds in two phases: The first phase
runs as long as there are prediction mandatory edges, i.e., edges that must be contained in every
feasible query set under the assumption that the predictions are correct; \camera{we later give a formal characterization of such edges}. In this phase, 
we exploit the existence of those edges and their properties to execute queries with strong local guarantees, i.e., 
each feasible query set contains a large portion of our queries.
For the second phase, we observe and exploit that the absence of prediction mandatory queries implies that the predicted optimal solution is a minimum vertex cover in a bipartite auxiliary graph.
The challenge here is that the auxiliary graph can change
with each wrong prediction. To obtain an error-dependent guarantee (our error
measure $k_h$ is discussed below) 
we need to
adaptively query a dynamically changing minimum vertex cover.

\medskip \noindent {\bf \nnew{Novel} techniques}\quad 
During the first phase, we 
\nnew{generalize} the classical witness set analysis.
In computing with explorable uncertainty,
the concept of {\em witness sets} is crucial for comparing the query set of an algorithm with an optimal solution (\nnew{a way of lower-bounding}). 
A 
witness set~\cite{bruce05uncertainty} is a set of elements for which we can guarantee that any feasible solution must query at least {\em one} of these elements.
Known algorithms for the MST problem without predictions~\cite{erlebach08steiner_uncertainty,megow17mst} essentially follow the algorithms of Kruskal or Prim and only identify witness sets of size $2$ in the cycle or cut that is currently under consideration. 
Querying disjoint witness sets of size~$2$ (witness pairs)
ensures $2$-robustness but does not lead to 
an improved consistency.
 
In our first phase, we extend this concept by considering \emph{strengthened} witness
sets of three elements such that any feasible query set must contain at least two of them.
Since we cannot always find strengthened witness sets based on structural
properties alone (otherwise, there would be a $1.5$-competitive algorithm for
the problem without predictions), we identify such sets under the assumption
that the predictions are correct.
Even after identifying such elements, the algorithm needs to query them in a careful order: \camera{if} the predictions are wrong, we lose the guarantee on the elements, and querying all of them might violate the robustness.
In order to identify strengthened witness sets, we provide new, more global criteria to identify additional witness sets (of size two)
beyond the current cycle or cut. 
These new ways to identify witness sets are a major contribution that may be of
independent interest regardless of predictions. 
During the first phase, we 
repeatedly query $\gamma-2$ prediction mandatory edges
together with a strengthened witness set, 
which
ensures $(1+\frac{1}{\gamma})$-consistency.
We query the elements in a carefully selected order while adjusting for errors to ensure $\gamma$-robustness.


For the second phase, we observe that the predicted optimal solution of the remaining instance is a minimum vertex cover $VC$ in a bipartite auxiliary graph representing the structure of \emph{potential} witness pairs (edges of the input graph correspond to vertices of the auxiliary graph).
For instances with this property, we aim for $1$-consistency and $2$-robustness; the best-possible trade-off for such instances.
If the predictions are correct, each edge of the auxiliary graph is a witness pair.
However, if a prediction error is observed when a vertex of $VC$ is queried, the auxiliary graph
changes.
This means that some edges of the original auxiliary graph are not actually witness pairs. Indeed, the size of a minimum vertex cover can increase and decrease and does not constitute a lower bound on $|\OPT|$; see \Cref{app:mst-optimal-tradeoff}.

If we only aim for consistency and robustness, we can circumvent this problem by selecting a distinct matching partner $h(e) \not\in VC$ for each $e \in VC$ 
applying \emph{K\H{o}nig-Egerv\'ary}'s Theorem (duality of maximum matchings and minimum vertex covers in bipartite graphs, see e.g.~\cite{Biggs1986}).
By querying the elements of $VC$ in a 
carefully chosen order until a prediction error is observed for the first time, we can guarantee that $\{e,h(e)\}$ is a witness set for each $e \in VC$ that is already queried.
In the case of an error, this allows} us to extend the previously queried elements to disjoint witness pairs to guarantee $2$-robustness.
Then, we can switch to an arbitrary (prediction-oblivious)
$2$-competitive algorithm for the remaining queries. 

If we additionally aim for an
error-sensitive guarantee, however, handling the dynamic changes to the auxiliary
graph, its minimum vertex cover $VC$ and matching $h$ requires substantial additional work,
and we see overcoming this challenge as our main contribution.
In particular, querying the partner $h(e)$ of each already queried $e \in VC$ in case of an error might be too expensive for the error-dependent guarantee. However, if we do not query these partners, the changed instance still depends on them, and if we charge against such a partner multiple times, we can lose the robustness.
Our solution is based on an elaborate
charging/counting scheme and involves:
\begin{itemize}
	\item keeping track of matching partners of already queried elements
of $VC$;
\item updating the matching and $VC$ using an augmenting path method to bound the number of elements that are charged against multiple times in relation to the 
prediction error;
\item and querying the partners of previously queried edges (and their new matching partners) as soon as they become endpoints of a newly matched edge, in order to prevent dependencies between the (only partially queried) witness sets of previously queried edges.
\end{itemize}

The error-sensitive algorithm achieves a 
competitive ratio
of $1+\frac{1}{\gamma}+\frac{5 k_h}{|\OPT|}$, at the price of a slightly increased
robustness of~$\gamma+1$ instead of $\gamma$.


\medskip \noindent {\bf Hop distance as error metric}\quad  
When we aim for a fine-grained performance analysis giving guarantees that depend on the quality of predictions, we need a metric to measure the prediction error. 
A very natural, simple error measure is the number of inaccurate predictions $k_\#=|\{e \in E\,|\, w_e \not= \w_e\}|$. However, we can show that even for $k_{\#} = 1$  the competitive ratio cannot be better than the known bound of~$2$ (see Lemma~\ref{lem:lb_wrong_predictions} in Appendix~\ref{app:prelim}). The reason for the weakness of this measure is that it completely ignores the interleaving  structure of intervals. Similarly, an $\ell_1$ error metric such as
$\sum_{e\in E}|w_e - \w_e|$ would not be meaningful because only
the \emph{order} of the values and the interval endpoints matters for our problems.

We propose a refined error measure, which we call {\em hop distance}. It is very intuitive even though it requires some technical care to make it precise. If we consider only a single predicted value $\w_e$ for some $e \in E$, then, in a sense, this value predicts the relation of the true value $w_e$ to the intervals of edges $e' \in E\setminus \{e\}$.
In particular, w.r.t.~a fixed $e' \in E\setminus \{e\}$, the value $\w_e$ predicts whether $w_e$ is left of $I_{e'}$ ($\w_e \le L_{e'}$), right of $I_{e'}$ ($\w_e \ge U_{e'}$), or contained in $I_{e'}$ ($L_{e'} < \w_e < U_{e'}$).
Interpreting the prediction $\w_e$ in this way, the prediction is \enquote{wrong} (w.r.t.~a fixed $e' \in E\setminus \{e\}$) if the predicted relation of the true value $w_e$ to interval $I_{e'}$ is not actually true, e.g., $w_e$ is predicted to be left of $I_{e'}$ ($\w_e \le L_{e'}$) but the actual $w_e$ is either contained in or right of $I_{e'}$ ($w_e > L_{e'}$).
Formally, we define the function $k_{e'}(e)$ that indicates whether the predicted relation of $w_e$ to $I_{e'}$ is true ($k_{e'}(e) = 0$) or not ($k_{e'}(e) = 1$).
With the prediction error $\jo(e)$ for a single $e \in E$, we want to capture the number of relations between $w_e$ and intervals $I_{e'}$ with $e' \in E\setminus \{e\}$ that are not accurately predicted.
Thus, we define $\jo(e) = \sum_{e' \in E\setminus \{e\}} k_{e'}(e)$.
For a set of edges $E' \subseteq E$, we define $\jo(E') = \sum_{e \in E'} \jo(e)$.
Consequently, with the error for the complete instance we want to capture the total number of wrongly predicted relations and, therefore, define it by $k_h=\jo(E)$.
We call this error measure $k_h$ the {\em hop distance}; see Figure~\ref{fig_error_ex} for an illustration.

Symmetrically, we can define $\oj(e) = \sum_{e' \in E\setminus \{e\}} k_{e}(e')$ and $\oj(E') = \sum_{e \in E'} \oj(e)$ for subset $E' \subseteq E$. 
Then $\jo(E) = k_h = \oj(E)$ follows by reordering the summations. 


\begin{figure}[t]
	\centering
	\begin{subfigure}[l]{0.2\textwidth}
	\scalebox{0.6}{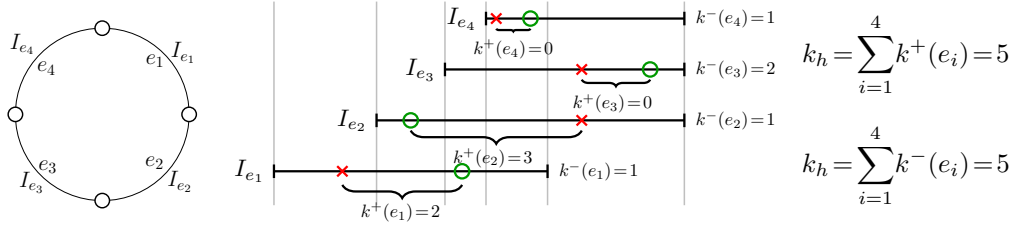}
	\end{subfigure}
	\begin{subfigure}[l]{0.25\textwidth}
		\begin{tikzpicture}[line width = 0.3mm, scale = 0.9, transform shape]	
		\intervalpr{$I_{e_1}$}{0}{4}{1.5}{1}{2.75}	
		\intervalpr{$I_{e_2}$}{1.5}{6}{2.25}{4.5}{2}	
		\intervalpr{$I_{e_3}$}{2.5}{6}{3}{4.5}{5.5}							
		\intervalpr{$I_{e_4}$}{3.1}{6}{3.75}{3.25}{3.75}		
		
		\begin{scope}[on background layer]
		\draw[line width = 0.2mm,,lightgray] (0,1) -- (0,4);
		\draw[line width = 0.2mm,,lightgray] (1.5,1) -- (1.5,4);
		\draw[line width = 0.2mm,,lightgray] (2.5,1) -- (2.5,4);
		\draw[line width = 0.2mm,,lightgray] (3.1,1) -- (3.1,4);
		\draw[line width = 0.2mm,,lightgray] (6,1) -- (6,4);
		\draw[line width = 0.2mm,,lightgray] (4,1) -- (4,4);
		\end{scope}
		
		\draw[decoration= {brace, amplitude = 5 pt, aspect = 0.5}, decorate] (2.75,1.3) -- (1,1.3);
		\node[] (l1) at (1.85,0.9) {\scalebox{0.75}{$\jo(e_1)=2$}};
		
		\draw[decoration= {brace, amplitude = 5 pt, aspect = 0.5}, decorate] (4.5,2.1) -- (2,2.1);
		\node[] (l1) at (3.2,1.7) {\scalebox{0.75}{$\jo(e_2)=3$}};
		
		\draw[decoration= {brace, amplitude = 2.5 pt, aspect = 0.5}, decorate] (5.5,2.85) -- (4.5,2.85);
		\node[] (l1) at (4.95,2.5) {\scalebox{0.75}{$\jo(e_3)=0$}};
		
		\draw[decoration= {brace, amplitude = 1 pt, aspect = 0.5}, decorate] (3.75,3.6) -- (3.25,3.6);
		\node[] (l1) at (3.5,3.3) {\scalebox{0.75}{$\jo(e_4)=0$}};

		\node[] (l1) at (4.75,1.5) {\scalebox{0.75}{$\oj(e_1)=1$}};
		\node[] (l1) at (6.75,2.25) {\scalebox{0.75}{$\oj(e_2)=1$}};
		\node[] (l1) at (6.75,3) {\scalebox{0.75}{$\oj(e_3)=2$}};
		\node[] (l1) at (6.75,3.75) {\scalebox{0.75}{$\oj(e_4)=1$}};
		\end{tikzpicture}
	\end{subfigure}
	\hspace*{3.5cm}
	\begin{subfigure}[r]{0.25\textwidth}
	\vspace*{-0.5cm}
	$$
	k_h = \sum_{i=1}^{4} \jo(e_i) = 5
	$$
	$$
	k_h = \sum_{i=1}^{4} \oj(e_i) = 5
	$$
	\end{subfigure}
	
%
\caption{
 Example of a single cycle (left) with uncertain edge weights from intersecting intervals  
$I_{e_1},I_{e_2},I_{e_3},I_{e_4}$ (middle). Circles illustrate true values and crosses illustrate the predicted values. 
}
\label{fig_error_ex}
\end{figure}

\section{Structural results}
\label{sec:mst:prelim}

In this section, we introduce some known concepts and prove new structural results on MST 
under explorable uncertainty, \tomc{which} we use later to design learning-augmented algorithms. 

\medskip \noindent \textbf{Witness sets and mandatory queries}\ \ Witness sets are the key to  the analysis of query algorithms. They allow for a comparison of an algorithm's query set to an optimal solution.
A subset $W \subseteq E$ is a \emph{witness set} if $W \cap Q \not= \emptyset$ for all feasible query sets $Q$.
An important special case are witness sets of cardinality one, i.e., edges that are part of every feasible query set.
We call such edges \emph{mandatory}. 
Similarly, we call edges that are mandatory under the assumption that the predictions are correct {\em prediction mandatory}.

\camera{For an example, consider Figure~\ref{fix-ex-mandatory}. In the example, we see the uncertainty intervals, predicted values and true values of four edges that form a simple cycle. We can observe that both $e_1$ and $e_2$ are mandatory for this example. 
To see this, assume that $e_1$ is not mandatory. Then, there must be a feasible query set $Q$ with $e_1 \not\in Q$ for the instance, which implies that $Q = \{e_2,e_3,e_4\}$ must be feasible.
But even after querying $Q$ to reveal the true values of $e_2$, $e_3$ and $e_4$, it still depends on the still unknown true value of $e_1$ whether there exists an MST $T$ with $e_1 \in T$ (only if $w_{e_1} \le w_{e_2}$) and/or $e_1 \not\in T$ (only if $w_{e_2} \le w_{e_1}$).
Even after querying $Q$ there is no spanning tree $T$ that is an MST for each possible edge weight in $I_{e_1}$ of the unqueried edge $e_1$ and, thus, $Q$ is not feasible.
This implies that $e_1$ is mandatory, and we can argue analogously that $e_2$ is mandatory as well. 

To argue whether an edge is prediction mandatory, on the other hand, we assume that all queries reveal the predicted values as true values. 
Under this assumption, a query to $e_1$ in the example would reveal a value that is larger than all upper limits $U_{e_i}$ with $i \in \{2,3,4\}$, which implies that $e_1$ cannot be part of any MST and that $T=\{e_2,e_3,e_4\}$ is an MST no matter what the true values of $e_2$, $e_3$ and $e_4$ actually are.
Therefore, under the assumption that all predicted values match the true values, $Q=\{e_1\}$ is a feasible query set and, thus, $e_2$ is not prediction mandatory despite being mandatory.
However, we can use a similar argumentation as above to argue that $e_1$ is also prediction mandatory.
}

\begin{figure}[tb]
	
		\centering
	\begin{subfigure}[l]{0.45\textwidth}
		
		\centering
		\scalebox{0.6}{\input{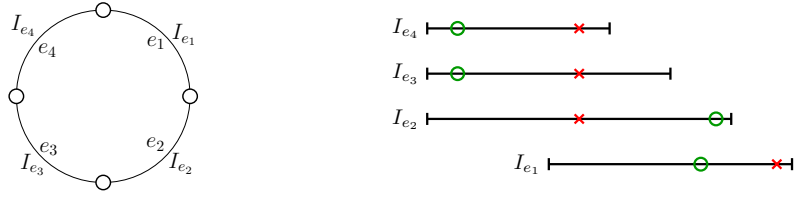}}
	\end{subfigure}
	\begin{subfigure}[r]{0.45\textwidth}
		\centering
		\begin{tikzpicture}[line width = 0.3mm, scale = 0.8, transform shape]
		\intervalpr{$I_{e_1}$}{2}{6}{1.5}{5.75}{4.5}

		\intervalpr{$I_{e_2}$}{0}{5}{2.25}{2.5}{4.75}
		
		\intervalpr{$I_{e_3}$}{0}{4}{3}{2.5}{0.5}						
		\intervalpr{$I_{e_4}$}{0}{3}{3.75}{2.5}{0.5}				
		
		\end{tikzpicture}
	\end{subfigure}
	\caption{
	 \camera{Example of a single cycle (left) with uncertain edge weights from intersecting intervals  
	$I_{e_1},I_{e_2},I_{e_3},I_{e_4}$ \camerate{(right)}. Circles illustrate true values and crosses illustrate the predicted values. For the example, $\{e_1,e_2\}$ is the set of all mandatory edges and $\{e_1\}$ is the set of all prediction mandatory edges.}
	}
	\label{fix-ex-mandatory}
\end{figure}


\camera{We continue by giving properties that allow us to identify (prediction) mandatory edges. To that end, let}
the \emph{lower limit tree} $T_L \subseteq E$ be an MST for
{values} $w^L$ with $w^L_e = L_e + \epsilon$ for an infinitesimally small $\epsilon > 0$. 
Analogously, let the \emph{upper limit tree} $T_U$ be an MST for {values}  $w^U$ with $w^U_e = U_e - \epsilon$.
This concept has been introduced in \cite{megow17mst} to identify mandatory queries; it is shown that 
every non-trivial $e \in T_L \setminus T_U$ is mandatory.
Thus, we may repeatedly query edges in 
$T_L \setminus T_U$ until $T_L = T_U$ without worsening the competitive ratio. 
We can extend this preprocessing to achieve uniqueness for $T_L$ and $T_U$ and, thus, may assume unique $T_L = T_U$.
	\begin{restatable}{lemma}{LemMSTPreprocessing}
		\label{mst_preprocessing}
		By querying only mandatory elements we can obtain an instance with $T_L = T_U$ such that $T_L$ and $T_U$ are the unique lower limit tree and upper limit tree, respectively.
	\end{restatable}
	Given an instance with unique $T_L=T_U$, we observe that each $e  \in T_L$ that is not part of the MST for the true values is mandatory.
	Similarly, each $e \not\in T_L$ that is part of the MST for the true values is mandatory as well. 
	A stronger version of this observation is as follows. 

	\begin{restatable}{lemma}{mstTreeChange}
		\label{lemma_mst_tree_change}
		Let $G$ be an instance with unique $T_L=T_U$ and let $G'$ be an instance with unique $T_L' = T_U'$ obtained from $G$ by querying set $Q$, then $e \in T_L \sym T_L' \jnew{= (T_L \setminus T'_L) \cup (T'_L \setminus T_L)}$ implies $e \in Q$.
	\end{restatable}
	

	Next we establish a relation between the set $E_M \subseteq E$ of mandatory queries, the set $E_P \subseteq E$ of prediction mandatory queries, and the hop distance $k_h$.
	
	\begin{restatable}{lemma}{lemHopDist}
		\label{Theo_hop_distance_mandatory_distance}
		Consider a given problem instance $G=(V,E)$ with predicted values $\pred{w}$.
		Each $e \in E_M \sym E_P$ satisfies $\oj(e) \ge 1$.
		Consequently, $k_h \ge |E_M \sym E_P|$.
	\end{restatable}

	For a formal proof, we refer to the appendix. 
		Intuitively, if $e \in E_P \setminus E_M$, then it is possible to solve the instance without querying $e$.
		Thus, the relation of the true values $w_{e'}$ with $e' \in E \setminus \{e\}$ to interval $I_e$ must be such that querying $E\setminus \{e\}$ allows us to verify that $e$ is either part or not part of the MST.
		If the predicted relations of the true values $w_{e'}$ with $e' \in E \setminus \{e\}$ to interval $I_e$ were the same, then querying $E\setminus \{e\}$  would also allow us to verify that $e$ is either part or not part of the MST. 
		Thus, the predicted relation of at least one $w_{e'}$ to interval $I_e$ must be wrong.
		We can argue analogously for $e \in E_M \setminus E_P$.

	\medskip \noindent \textbf{Identifying witness sets}\ \ 
	%
	We introduce new structural properties to identify witness sets.
	Existing algorithms for MST under uncertainty~\cite{erlebach08steiner_uncertainty,megow17mst} essentially follow the algorithms of Kruskal or Prim, and only identify witness sets in the cycle or cut that is currently under consideration.
	Let $f_1,\ldots,f_l$ denote the edges in $E\setminus T_L$ ordered by non-decreasing lower limit.
	Then, $C_i$ with $i \in \{1,\ldots,l\}$ denotes the unique cycle in $T_L \cup \{f_i\}$.

	For each $e \in T_L$, let $X_e$ denote the set of edges in the cut of the two connected components of $T_L \setminus \{e\}$.
	Existing algorithms for MST under explorable uncertainty repeatedly consider (the changing) $C_1$ or $X_e$, where $e$ is the edge in $T_L$ with maximum upper limit, and identify the maximum or minimum edge in the cycle or cut by querying witness sets of size two, until the problem is solved. 
	For our algorithms, we need to identify witness sets in cycles $C_i \not= C_1$ and cuts $X_{e'} \not= X_e$.

	\begin{restatable}{lemma}{LemMstWitnessFir}
		\label{lemma_mst_witness_set_1}		
		Consider cycle $C_i$ with $i \in \{1,\ldots,l\}$. 
		Let $l_i \in C_i \setminus\{f_i\}$ such that $I_{l_i} \cap I_{f_i} \not= \emptyset$ and~$l_i$ has the largest upper limit in $C_i \setminus \{f_i\}$, then $\{f_i,l_i\}$ is a witness set. 
		If $w_{f_i} \in I_{l_i}$, then $l_i$ is mandatory.
	\end{restatable}

\noindent {\bf Characterization of prediction mandatory free instances\ \ }
We say an instance is \emph{prediction mandatory free} if it contains no prediction mandatory elements.
A key part of our algorithms is to transform instances into prediction mandatory free instances while maintaining a competitive ratio that
allows us to achieve the optimal consistency and robustness trade-off overall.
We give the following characterization of prediction mandatory free instances, (cf.~Figure~\ref{Ex_mst_phase1_cases}).
Then, we show that prediction mandatory free instances remain so as long as we ensure $T_L = T_U$.
	
	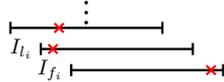
\begin{figure}[th]
			\centering
			\begin{tikzpicture}[line width = 0.3mm, scale = 0.8, transform shape]
			\intervalp{$I_{f_i}$}{-3}{-0.5}{0.3}{-0.7}
			\intervalp{$I_{l_i}$}{-3.5}{-1}{0.65}{-3.3}
			\intervalp{}{-4}{-1.5}{1}{-3.2}
			\path (-2.75, 1.25) -- (-2.75, 1.25) node[font=\LARGE, midway, sloped]{$\dots$};
			\end{tikzpicture}
		\caption{Intervals in a prediction mandatory free cycle with predictions indicated as red crosses  
		}
		\label{Ex_mst_phase1_cases}
	\end{figure}
	
	\begin{restatable}{lemma}{LemMSTPredFreeIff}
		\label{mst_pred_free_characterization}	\label{lem:remaining_predfree}
		{
			An instance $G$ is prediction mandatory free iff  
			$\pred{w}_{f_i} \ge U_{e}$ and $\pred{w}_e \le L_{f_i}$ holds for each $e \in C_i \setminus \{f_i\}$ and each cycle $C_i$ with ${i} \in \{1,\ldots,l\}$.}
		Once an instance is prediction mandatory free, it remains so even if we query further elements, as long as we maintain unique $T_L = T_U$.
	\end{restatable}

\noindent {\bf \nnew{
	Making instances prediction mandatory free}\ \ }
	In Appendix~\ref{app:mst-predfree}, we give a \nnew{powerful} preprocessing algorithm that transforms arbitrary instances into prediction mandatory free instances. 

	\begin{restatable}{theorem}{MSTEndOfPhaseOne}
		\label{mst_end_of_phase_one}
		There is an algorithm that makes a given instance prediction mandatory free and satisfies $|\ALG| \le \min\{ (1 + \frac{1}{\gamma}) \cdot (|(\ALG \cup D) \cap \OPT| + \jo(\ALG) + \oj(\ALG)), \gamma \cdot |(\ALG \cup D) \cap \OPT|  + \gamma -2\}$ for the set of edges $\ALG$ queried by \nnew{the algorithm} 
		\jnew{and a set $D \subseteq E\setminus \ALG$ of unqueried edges that do not occur in the remaining instance \emph{after} executing the algorithm.}
	\end{restatable}

	\jnew{
	The set $D$ are edges that, even without being queried by the algorithm, are proven to be maximal in a cycle or minimal in a cut. 
	Thus, they can be deleted or contracted 
	\nnew{w.l.o.g.}\ and do not 
	\nnew{exist} in the instance remaining \emph{after} executing the preprocessing algorithm.
	This is an important property as it means that the remaining instance is independent of $D$ and $\ALG$ (as all elements of $\ALG$ are already queried).
	Since the theorem compares $|\ALG|$ with $|(\ALG \cup D) \cap \OPT|$ instead of just $|\OPT|$, this allows us to combine the \nnew{given} guarantee 
		with the guarantees of dedicated algorithms for prediction mandatory free instances.
	However, we have to be careful with the additive term 
	$\gamma-2$, 
	but we will see that we can charge this term against the improved robustness of our algorithms for prediction mandatory free instances.}


\section{\nnew{An algorithm with} optimal consistency and robustness trade-off}
\label{sec:optimal-tradeoff}
We give a bound on the best~achievable~tradeoff between consistency and robustness. 

\begin{restatable}{theorem}{ThmLBTradeoffWithoutError}
\label{theo_minimum_combined_lb}
  Let 
  $\beta \geq 2$ be a fixed integer.
  For the MST problem \camera{under explorable uncertainty with predictions}, there is no deterministic $\beta$-robust algorithm that is $\alpha$-consistent for $\alpha < 1 + \frac{1}{\beta}$. And vice versa, no deterministic $\alpha$-consistent algorithm, with $\alpha>1$, is $\beta$-robust for $\beta < \max\{\frac{1}{\alpha-1},2\}$.
\end{restatable}

The main result of this section is an optimal algorithm w.r.t.\ this tradeoff bound. 

\begin{restatable}{theorem}{optTradeoff}
	\label{thm:optimal-tradeoff}
	For every integer $\gamma\ge 2$, there exists a $(1+\frac{1}{\gamma})$-consistent and $\gamma$-robust algorithm for the MST problem \camera{under explorable uncertainty with predictions}.
\end{restatable}


\nnew{To show this result, we design an algorithm for prediction mandatory free instances with unique $T_L = T_U$. We run it after the preprocessing algorithm which obtains such special instance with the query guarantee in~\Cref{mst_end_of_phase_one}. Our new algorithm achieves the optimal~trade-off.}

\begin{theorem}
	\label{thm:predfree}
	There exists a $1$-consistent and $2$-robust algorithm for prediction mandatory free instances with unique $T_L = T_U$ of the MST problem \camera{under explorable uncertainty with predictions}.
\end{theorem}

In a prediction mandatory free instance $G=(V,E)$, each $f_i \in E \setminus T_L$ is predicted to be maximal on cycle $C_i$, and each $l \in T_L$ is predicted to be minimal in $X_l$ (cf.~\Cref{mst_pred_free_characterization}). 
{If these predictions are correct, then $T_L$ is an MST and the optimal query set is a minimum vertex cover in a bipartite graph $\bar{G} = (\bar{V},\bar{E})$ with $\bar{V} = E$ (excluding trivial edges) and $\bar{E} = \{ \{f_i,e\} \mid i \in \{1,\ldots,l\}, e \in C_i \setminus \{f_i\} \text{ and } I_e \cap I_{f_i} \not= \emptyset\}$ \cite{erlebach14mstverification,megow17mst}. We refer to $\bar{G}$ as the {\em vertex cover instance}.}
Note that {if a query reveals that} an $f_i$ is not maximal on $C_i$ or an $l \in T_L$ is not minimal in $X_{l}$, then the vertex cover instance changes.
Let $VC$ be a minimum vertex cover of $\bar{G}$.
Non-adaptively querying $VC$ ensures $1$-consistency but might lead to an arbitrar\nnew{ily} bad robustness. Indeed, the size of a minimum vertex cover can increase and decrease drastically as 
shown in \Cref{app:mst-optimal-tradeoff}.
Thus, the algorithm has to be more adaptive.

The idea of the algorithm (cf. Algorithm~\ref{ALG_mst_part_2}) is to sequentially query each $e \in VC$ and charge for querying $e$ by a distinct non-queried element $h(e)$ such that $\{e,h(e)\}$ is a witness set.
Querying exactly one element per 
{distinct} witness set implies optimality.
To identify~$h(e)$ for each element $e \in VC$, we use  {the fact that
\emph{K\H{o}nig-Egerv\'ary}'s Theorem (e.g.,~\cite{Biggs1986}) on the duality between minimum vertex covers and maximum matchings in bipartite graphs} implies that there is a matching $h$ that maps each $e \in VC$ to a distinct $e' \not\in VC$. 
While the sets $\{e,h(e)\}$ with $e \in VC$ in general are not witness sets, querying $VC$ in a specific order until the vertex cover instance changes guarantees that $\{e,h(e)\}$ is a witness set for each already queried~$e$.  
The algorithm queries in this order until it detects a wrong prediction or solves the problem.
If it finds a wrong prediction, it queries the partner $h(e)$ of each already queried edge $e$, and continues to solve the problem with a $2$-competitive algorithm (e.g.,~\cite{erlebach08steiner_uncertainty,megow17mst}).
The following lemma specifies the order in which the algorithm queries 
$VC$.

\begin{restatable}{lemma}{MstPhaseTwTwo}
	\label{mst_phase2_2}
	Let $l'_1,\ldots,l'_k$ be the edges in $VC \cap T_L$ ordered by non-increasing upper limit and let $d$ be such that \camera{the true value of} each $l'_{i}$ with $i < d$ is minimal in cut $X_{l'_i}$, then $\{l'_{i},h(l'_{i})\}$ is a witness set for each $i \le d$. Let $f'_1,\ldots,f'_g$ be the edges in $VC \setminus T_L$ ordered by non-decreasing lower limit and let $b$ be such that \camera{the true value of} each $f'_{i}$ with $i < b$ is maximal in cycle $C_{f'_i}$, then $\{f'_{i},h(f'_{i})\}$ is a witness set for each $i \le b$.
\end{restatable}

\begin{proof}
	Here, we show the first statement and refer to Appendix~\ref{app:mst-optimal-tradeoff} for the proof of the second statement.
	Consider an arbitrary $l'_i$ and $h(l'_i)$ with $i \le d$. 
	By definition of $h$, the edge $h(l'_i)$ is not part of the lower limit tree.
	 {Consider $C_{h(l'_i)}$, i.e., the cycle in $T_L \cup \{h(l_i')\}$, then we claim that $C_{h(l'_i)}$ only contains $h(l'_i)$ and edges in $\{l'_{1},\ldots,l'_{k}\}$ (and possibly irrelevant edges that do not intersect $I_{h(l'_i)}$). 
	To see this, recall that $l_i' \in VC$, by definition of $h$, implies $h(l_i') \not\in VC$.
	For $VC$ to be a vertex cover, each $e \in C_{h(l_i')}\setminus\{h(l_i')\} $ must either be in $VC$ or not intersect $h(l_i')$.}
	Consider the relaxed instance where the true values for each $l'_j$ with $j<d$ and $j \not= i$ are already known. 
	By assumption each such $l'_j$ is minimal in its cut $X_{l'_j}$.
	Thus, we can w.l.o.g.\ contract each such edge.
	It follows that in the relaxed instance $l'_i$ has the highest upper limit in $C_{h(l'_i)} \setminus \{h(l'_i)\}$.
	According to Lemma~\ref{lemma_mst_witness_set_1}, $\{l'_i, h(l'_i)\}$ is a witness set.
\end{proof}

%

\begin{algorithm}[tb]
	\KwIn{Prediction mandatory free graph $G=(V,E)$ with unique $T_L = T_U$.} 
	Compute maximum matching $h$ and minimum vertex cover $VC$ for $\bar{G}$\label{line_mst_two_compute_vc}\; 
	Set $W = \emptyset$, and let $f'_1,\ldots,f'_g$ and $l'_1,\ldots,l'_k$ be as described in Lemma~\ref{mst_phase2_2}\label{line_mst_two_l_edges}\;
	\For{$e$ chosen sequentially from the {ordered list} $f'_1,\ldots,f'_g,l'_1,\ldots,l'_k$\label{line_mst_two_main_for}}{
		Query $e$\label{line_mst_two_query_vc} and add $h(e)$ to $W$\; 
		\lIf{$\jo(e)\not=0$}{ query set $W$\label{line_mst_two_rob1} and solve the 
		instance with a $2$-competitive algorithm}
	}
	\caption{$1$-consistent and $2$-robust algorithm for prediction mandatory free instances.}
	\label{ALG_mst_part_2}
\end{algorithm}

\vspace*{-2ex}
\begin{proof}[Proof of Theorem~\ref{thm:predfree}]
	We first show $1$-consistency.
	Assume that all predictions are correct, then $VC$ is an optimal query set and $\jo(e)=0$ holds for all $e\in E$.
	It follows that Line~\ref{line_mst_two_rob1} never executes queries and the algorithm queries exactly $VC$. This implies $1$-consistency.
	
	\nnew{Further,} if the algorithm never queries in Line~\ref{line_mst_two_rob1}, then the consistency analysis implies $1$-robustness.
	\nnew{Suppose} Line~\ref{line_mst_two_rob1} executes queries.
	Let $Q_1$ denote the set of edges that are queried before the queries of Line~\ref{line_mst_two_rob1} and let $Q_2 = \{h(e) \mid e \in Q_1\}$. Then $Q_2$ corresponds to the set $W$ as queried 
	in Line~\ref{line_mst_two_rob1}. 
	By Lemma~\ref{mst_phase2_2}, each $\{e,h(e)\}$ with $e \in Q_1$ is a witness set.
	Further, the sets $\{e,h(e)\}$ are pairwise disjoint.
	Thus, $|Q_1 \cup Q_2| \le 2 \cdot |\OPT \cap (Q_1 \cup Q_2)|$.
	Apart from $Q_1 \cup Q_2$, the algorithm queries a set $Q_3$ in Line~\ref{line_mst_two_rob1} to solve the remaining instance with a $2$-competitive algorithm. 
	So, $|Q_3| \le 2 \cdot |\OPT \setminus (Q_1\cup Q_2)|$ and, adding up the inequalities, $|\ALG| \le 2 \cdot |\OPT|$.
\end{proof}

\jnew{A careful combination of~\Cref{thm:predfree,mst_end_of_phase_one} implies~\Cref{thm:optimal-tradeoff}. 
Full proof in Appendix~\ref{app:mst-optimal-tradeoff}}.

\section{An error-sensitive algorithm}
\label{sec:error-sensitive}
In this section, we extend the algorithm of~\Cref{sec:optimal-tradeoff} to obtain error sensitivity.
First, we note that $k_{\#} = 0$ implies $k_h = 0$, so Theorem~\ref{theo_minimum_combined_lb} implies that no algorithm can simultaneously have competitive ratio better than $1 + \frac{1}{\beta}$ if $k_h = 0$ and $\beta$ for arbitrary~$k_h$.	
In addition, we can give the following lower bound on the competitive ratio
as a function of $k_{h}$.

\begin{restatable}{theorem}{ThmLBAllErrorMeasures}
\label{thm_lb_error_measure}
	Any deterministic algorithm for MST \camera{under explorable uncertainty with predictions} has a competitive ratio $\rho\geq \min\{1+\frac{k_h}{\opt},2\}$, 
	even for edge disjoint prediction mandatory free cycles.
\end{restatable}

Again, we design an algorithm for prediction mandatory free instances with unique $T_L = T_U$ and use it in combination with 
\nnew{the preprocessing algorithm (\Cref{mst_end_of_phase_one})} to prove the following. 

\begin{restatable}{theorem}{thmErrorSensitive}
	\label{thm:error-sensitive}
	For every integer $\gamma\ge 2$, there exists a $\min\{1+ \frac{1}{\gamma} + \frac{\red{5} \cdot k_h}{|\OPT|}, \gamma+1 \}$-competitive algorithm for the MST problem \camera{under explorable uncertainty with predictions}.
\end{restatable}

We actually show a robustness of $\max\{3,\gamma + \frac{1}{|\OPT|}\}$ which might be smaller than $\gamma +1$.
Our algorithm for prediction mandatory free instances asymptotically matches the error-dependent guarantee of~\Cref{thm_lb_error_measure} at the cost of a slightly worse robustness.

\begin{theorem}
	\label{thm:predfree-error}
	There exists a $\min\{1+ \frac{5 \cdot k_h}{|\OPT|}, 3\}$-competitive algorithm for prediction mandatory free instances with unique $T_L = T_U$ of the MST problem \camera{under explorable uncertainty with predictions}.
\end{theorem}

We follow the same strategy as before.
However, Algorithm~\ref{ALG_mst_part_2} just executes a 
$2$-competitive algorithm once it detects an error. 
This is sufficient to achieve the optimal trade-off as we, if an error occurs, only have to guarantee $2$-competitiveness.
To obtain an error-sensitive guarantee however, we have to ensure both, $|\ALG| \le 3 \cdot |\OPT|$ \emph{and} $|\ALG| \le \OPT + 5 \cdot k_h$ even if errors occur. 
Further, we might not be able to afford queries to the complete set $W$ (Algorithm~\ref{ALG_mst_part_2}, Line~\ref{line_mst_two_rob1}) in the case of an error as this might violate $|\ALG| \le \OPT + 5 \cdot k_h$.

We adjust the algorithm to query elements of $f'_1,\ldots,f'_g$ and $l'_1,\ldots,l'_k$ as described in Lemma~\ref{mst_phase2_2} not only until an error occurs but until the vertex cover instance changes. 
That is, until some $f_i$ that at the beginning of the iteration is not part of $T_L$ becomes part of the lower limit tree, or some $l_i$ that at the beginning of the iteration is part of $T_L$ is not part of the lower limit tree anymore.
Once the instance changes, we recompute both, the bipartite graph $\bar{G}$ as well as the matching $h$ and minimum vertex cover $VC$ for $\bar{G}$.
Instead of querying the complete set $W$, we only query the elements of $W$ that occur in the recomputed matching, as well as the new matching partners of those elements.
And instead of switching to a $2$-competitive algorithm, we restart the algorithm with the recomputed matching and vertex cover.
Algorithm~\ref{ALG_mst_part_2_error} formalizes this approach.
In the algorithm, $h$ denotes a matching that matches each $e \in VC$ to a distinct $h(e) \not\in VC$; we use the notation $\{e,e'\} \in h$ to indicate that $h$ matches $e$ and $e'$. 
For a subset $U \subseteq VC$ let $h(U) = \{h(e) \mid e \in U\}$.
For technical reasons, the algorithm does not recompute an arbitrary matching $h$ but follows the approach of Lines~\ref{line_mst_two_error_rematch1} and~\ref{line_mst_two_error_rematch2}.
Intuitively, an arbitrary maximum matching $h$ 
might contain too many elements of $W$, which would lead to too many additional queries. 

\begin{algorithm}[tb]
	\KwIn{Prediction mandatory free graph $G=(V,E)$ with unique $T_L = T_U$. 
	}
	Compute maximum matching $h$ and minimum vertex cover $VC$ for $\bar{G}$ and set $W= \emptyset$\label{line_mst_two_error_compute_vc}\; 
	Let $f'_1,\ldots,f'_g$ and $l'_1,\ldots,l'_k$ be as described in Lemma~\ref{mst_phase2_2}\label{line_mst_two_error_l_edges} for $VC$ and $h$\;
	$L \gets T_L$, $N \gets E \setminus T_L$ \tcc*{\camera{recompute actual $T_L$ (and $T_U$) after each query}}
	\For{$e$ chosen sequentially from the {ordered list} $f'_1,\ldots,f'_g,l'_1,\ldots,l'_k$\label{line_mst_two_error_main_for}}{
		If $e$ is non-trivial\camera{, i.e., has not been queried yet}, query $e$\label{line_mst_two_error_query_vc} and add $h(e)$ to $W$\; 
		\camera{Apply Lemma~\ref{mst_preprocessing} to} ensure unique $T_L = T_U$. For each query $e'$, if $\exists a \text{ s.t.\ } \{e',a\} \in h$, query $a$ \label{line_mst_two_error_man1}\;	
		Let $\bar{G}'=(\bar{V}',\bar{E}')$ be the vertex cover instance for the current instance\;
		\If{some $e'\in L$ is not in $T_L$ or some $e' \in N$ is in $T_L$}{
			\Repeat{$R = \emptyset$}{
				Let $\bar{G} = \bar{G}'$ and $\bar{h} = \{\{e',e''\} \in h \mid \{e',e''\} \in \bar{E}'\}$\label{line_mst_two_error_rematch1}\;
				Compute $h$ and $VC$ by completing $\bar{h}$ with an augmenting path algorithm\label{line_mst_two_error_rematch2}\;
				Query $R = (VC \cup h(VC)) \cap (W \cup \{e' \mid \exists e \in W \text{ with } \{e,e'\} \in h \})$ \label{line_mst_two_error_rob1}\;
				Ensure unique $T_L = T_U$. For each query $e'$, if $\exists a \text{ s.t.\ } \{e',a\} \in h$, query $a$\label{line_mst_two_error_man2}\;
				Let $\bar{G}'=(\bar{V}',\bar{E}')$ be the vertex cover instance for the current instance\label{line_mst_two_previous_vc}\;
			}
			Restart at Line~\ref{line_mst_two_error_l_edges}. \camera{In particular, do \emph{not} reset $W$}\;
		}
	}
	\caption{Error-sensitive algorithm for prediction mandatory free instances.
	}
	\label{ALG_mst_part_2_error}
\end{algorithm}

Let $\ALG$ denote the queries of Algorithm~\ref{ALG_mst_part_2_error} on a prediction mandatory free instance with unique $T_L = T_U$.
We show~\Cref{thm:predfree-error} by proving $\ALG \le \OPT + 5 \cdot k_h$ and $\ALG \le 3 \cdot \OPT$.

Before proving the two inequalities, we state some key observations about the algorithm.
We argue that an element $e'$ can never be part of a partial matching $\bar{h}$ in an execution of Line~\ref{line_mst_two_error_rematch1} \emph{after} it is added to set $W$. 
Recall that the vertex cover instances only contain non-trivial elements.
Thus, if an element $e$ is queried in Line~\ref{line_mst_two_error_query_vc} and the current partner $e' = h(e)$ is added to set $W$, then the vertex cover instance at the next execution of Line~\ref{line_mst_two_error_rematch1} does not contain the edge $\{e,e'\}$ and, therefore, $e'$ is not part of the partial matching $\bar{h}$ of that line.
As long as $e'$ is not added to the matching by Line~\ref{line_mst_two_error_rematch2}, it, by definition, can never be part of a partial matching $\bar{h}$ in an execution of Line~\ref{line_mst_two_error_rematch1}.
As soon as the element $e'$ is added to the matching in some execution of Line~\ref{line_mst_two_error_rematch2}, it is queried in the following execution of Line~\ref{line_mst_two_error_rob1}.
Therefore, $e'$ can also not be part of a partial matching $\bar{h}$ in an execution of Line~\ref{line_mst_two_error_rematch1} after it is added to the matching again.
This leads to the following observation.

\begin{obs}
\label{obs:hstar}
An element $e'$ can never be part of a partial matching $\bar{h}$ in an execution of Line~\ref{line_mst_two_error_rematch1} \emph{after} it is added to set $W$.
Once $e'$ is added to the matching again in an execution of Line~\ref{line_mst_two_error_rematch2}, it is queried directly afterwards in Line~\ref{line_mst_two_error_rob1}, and cannot occur in Line~\ref{line_mst_two_error_query_vc} anymore.
\end{obs}

We first analyze the queries that are \emph{not} executed in Line~\ref{line_mst_two_error_rob1}. 
Let $Q_1 \subseteq \ALG$ denote the queries of Line~\ref{line_mst_two_error_query_vc}.
For each $e \in Q_1$ let $h^*(e)$ be the matching partner of $e$ at the time it was queried, and let $h^*(Q_1) = \bigcup_{e \in Q_1} \{h^*(e)\}$.
Finally, let $Q_2$ denote the queries of Lines~\ref{line_mst_two_error_man1} and~\ref{line_mst_two_error_man2} to elements of $h^*(Q_1)$, and let $Q_3$ denote the remaining queries of those lines.

\begin{lemma}
\label{lem:nonrematchqueries}
$|Q_1 \cup Q_3 \cup h^*(Q_1)| \le 2 \cdot |\OPT \cap (Q_1 \cup Q_3 \cup h^*(Q_1))|$ and $|Q_1 \cup Q_2 \cup Q_3| \le  |\OPT \cap (Q_1 \cup Q_3 \cup h^*(Q_1))| + \oj(Q_2 \cup Q_3)$.  
\end{lemma}

\begin{proof}
	First, consider $Q_1$ and $h^*(Q_1)$.
	By~\Cref{lem:remaining_predfree}, the instance is prediction mandatory free at the beginning of each restart of the algorithm. 
	\nnew{By}~\Cref{mst_phase2_2}, 
	each $\{e,h^*(e)\}$ with $e \in Q_1$ is a witness set. 
	We claim that all such $\{e,h^*(e)\}$ are pairwise disjoint, which 
	implies $|Q_1 \cup h^*(Q_1)| \le 2 \cdot |\OPT \cap (Q_1 \cup h^*(Q_1))|$.
	Otherwise, an element of $\{e,h^*(e)\}$ must occur a second time in Line~\ref{line_mst_two_error_query_vc} after $e$ is queried and $h^*(e)$ is added to $W$.
	Thus, either $e$ or $h^*(e)$ must become part of a recomputed matching in line~\ref{line_mst_two_error_rematch1}.
	By \Cref{obs:hstar} and since $e$ becomes trivial, this cannot happen.

	Consider an $e \in Q_2 \subseteq h^*(Q_1)$ and let $e' \in Q_1$ with $h^*(e')  = e$.
	Since $e' \in Q_1$, it was queried in Line~\ref{line_mst_two_error_query_vc}.
	Observe that $e$ must have been queried after $e'$, as otherwise either $e'$ would not have been queried in Line~\ref{line_mst_two_error_query_vc} (but together with $e$ in Line~\ref{line_mst_two_error_man1} or~\ref{line_mst_two_error_man2}), or $e$ would not have been the matching partner of $e'$ when it was queried by \Cref{obs:hstar}; both contradict $e' \in Q_1$ and $h^*(e')  = e$.
	This and \Cref{obs:hstar} imply that, at the time $e$ is queried, its current matching partner is either the trivial $e'$ or it has no partner. 
	So, $e$ must have been queried because it was mandatory and not as the matching partner of a mandatory element.
	Thus, each query of $Q_2$ is mandatory but, by~\Cref{lem:remaining_predfree},  not prediction mandatory at the beginning of the iteration in which it is queried. 
	Therefore, \Cref{Theo_hop_distance_mandatory_distance} implies that all mandatory elements $e$ of $Q_2$ have $\oj(e)\ge 1$. 
	It follows $|Q_1 \cup Q_2| \le |\OPT \cap (Q_1 \cap h^*(Q_1))| + \oj(Q_2)$.

	By the argument above, no element of $Q_3$ was queried as the matching partner to an element of $Q_2 \cup Q_1$. 
	Thus, by~\Cref{mst_preprocessing} and the definition of the algorithm, at least half the elements of $Q_3$ are mandatory, and we have $|Q_3| \le 2 \cdot |\OPT \cap Q_3|$ (and $\frac{1}{2}|Q_3| \le |\OPT \cap Q_3|$), which implies $|Q_1 \cup Q_3 \cup h^*(Q_1)| \le 2 \cdot |\OPT \cap (Q_1 \cup Q_3 \cup h^*(Q_1))|$.

	By the same argument as for 
	$Q_2$,
	all mandatory elements $e$ of $Q_3$ have $\oj(e)\ge 1$.
	Thus, $\oj(Q_3) \ge \frac{1}{2} \cdot |Q_3|$. 
	Combining $\oj(Q_3) \ge \frac{1}{2} \cdot |Q_3|$ and $\frac{1}{2}|Q_3| \le |\OPT \cap Q_3|$ implies $|Q_3| \le |\OPT \cap Q_3| + \oj(Q_3)$.
	So, $|Q_1 \cup Q_2 \cup Q_3| \le |\OPT \cap (Q_1 \cup Q_3 \cup h^*(Q_1))|+ \oj(Q_2 \cup Q_3)$.
\end{proof}

The first part of \Cref{lem:nonrematchqueries} 
captures all queries 
outside of
Line~\ref{line_mst_two_error_rob1} and all queries of Line~\ref{line_mst_two_error_rob1} to elements of $W = h^*(Q_1)$. 
Let $Q'_4$ be the remaining queries of Line~\ref{line_mst_two_error_rob1}.
By definition of the algorithm, $|Q'_4| \le |W|$.
Since $|W| \le |\OPT|$, we can conclude the next lemma.

\begin{lemma}
	\label{lem:error_rob}
	$|\ALG| \le 3 \cdot |\OPT|$.
\end{lemma}

Next, we show $|\ALG| \le |\OPT| + 5 \cdot k_h$.
\Cref{lem:nonrematchqueries} implies
$|Q_1 \cup Q_2 \cup Q_3| \le  |\OPT \cap (Q_1 \cup Q_3 \cup h^*(Q_1))| + \oj(Q_2 \cup Q_3)$.
Hence, it remains to upper bound $|Q_4|$ with $Q_4 = \ALG \setminus (Q_1 \cup Q_2 \cup Q_3)$ by $4 \cdot k_h$.
By definition, $Q_4$ only contains edges that are queried in Line~\ref{line_mst_two_error_rob1}.
Thus, at least half the queries of $Q_4$ are elements of $W$ that are part of the matching $h$.
By~\Cref{obs:hstar}, no element of $W$ is part of the partial matching $\bar{h}$ in Line~\ref{line_mst_two_error_rematch1}.
Hence, in each execution of Line~\ref{line_mst_two_error_rob1}, at least half the queries are not part of $\bar{h}$ in  Line~\ref{line_mst_two_error_rematch1} but added to $h$ in Line~\ref{line_mst_two_error_rematch2}. 
Our goal is to bound the number of such elements.  

We start with some definitions.
Define $G_j$ as the problem instance at the $j$'th execution of Line~\ref{line_mst_two_error_rematch2}, and let $G_0$ denote the initial problem instance.
For each $G_j$, let $\bar{G}_j = (\bar{V}_j, \bar{E}_j)$, $T_L^j$ and $T_U^j$ denote the corresponding vertex cover instance and lower and upper limit trees.
Observe that each $G_j$ has unique $T^j_L = T^j_U$, and, by~\Cref{lem:remaining_predfree}, is prediction mandatory free.
Let $Y_j$ denote the set of queries made by the algorithm  to transform instance $G_{j-1}$ into instance $G_j$.
We partition $Q_{4}$ into subsets $S_j$, where $S_j$ contains the edges of $Q_{4}$ that are queried in the $j$'th execution of Line~\ref{line_mst_two_error_rob1}.
We claim 
$|S_j| \le 4 \cdot \jo(Y_j)$
for each $j$, which implies $|Q_{4}| \le \sum_j |S_j| \le  4 \cdot \sum_j \jo(Y_j) \le 4 \cdot k_h$.
To show the claim, we rely on the following
 lemma.

\begin{restatable}{lemma}{remEdgeError}
	\label{lem:oldedgeerror}
		Let $l,f$ be non-trivial \camera{edges} in $G_j$ such that  $\{l,f\} \in \bar{E}_{j-1} \sym \bar{E}_{j}$, then $\oj(l),\oj(f)\! \ge \! 1$.
		\jnew{Furthermore, $\jo(Y_j) \ge |U|$ for the set $U$ of all endpoints of such edges $\{l,f\}$.}
\end{restatable}

\begin{lemma}\label{lem:error-consistency}
	$|\ALG| \le |\OPT| + 5 \cdot k_h$.
\end{lemma}

\begin{proof}
	We show $|S_j| \le 4 \cdot \jo(Y_j)$ for each $j$, which, in combination with \Cref{lem:nonrematchqueries}, implies the lemma.
	Consider an arbitrary $S_j$ and the corresponding set $Y_j$.
	Further, let $h_{j-1}$ denote the maximum matching computed and used by the algorithm for vertex cover instance $\bar{G}_{j-1}$, and let  $\bar{h}_{j-1} = \{\{e,e'\} \in h_{j-1} \mid \{e,e'\} \in \bar{E}_j\}$.
	Finally, let $h_j$ denote the matching that the algorithm uses for vertex cover instance $\bar{G}_j$.
	That is, $h_j$ is computed by completing $\bar{h}_{j-1}$ with a standard augmenting path algorithm.
	As already argued, at least half the elements of $S_j$  are not matched by $\bar{h}_{j-1}$ but are matched by $h_j$ (cf.~\Cref{obs:hstar}).
	
	We bound the number of such elements by exploiting that $h_j$ is constructed from $\bar{h}_{j-1}$ via a standard augmenting path algorithm.
	By definition, each iteration of the augmenting path algorithm increases the size of the current matching (starting with $\bar{h}_{j-1}$) by one and, in doing so, matches two new elements.
	In total, at most $2 \cdot (|h_j|-|\bar{h}_{j-1}|)$ previously unmatched elements become part of the matching.
	Thus, $|S_j| \le 4 \cdot (|h_j|-|\bar{h}_{j-1}|)$.
	
	We show $(|h_j|-|\bar{h}_{j-1}|) \le  \jo(Y_j)$.
	According to \emph{K\H{o}nig-Egerv\'ary}'s Theorem (e.g.,~\cite{Biggs1986}), the size of $h_j$ is equal to the size $|VC_j|$ of the minimum vertex cover for~$\bar{G}_j$. 
	We show $|VC_j| \le |\bar{h}_{j-1}| + \jo(Y_j)$, which implies $(|h_j|-|\bar{h}_{j-1}|) {=}  |VC_j| - |\bar{h}_{j-1}| \le \jo(Y_j)$, and, thus, the claim. 	
	Let $\overline{VC}_{j-1} = \{e \in VC_{j-1} \mid \exists e' \text{ s.t.\ } \{e,e'\} \in \overline{h}_{j-1}\}$, then $|\overline{VC}_{j-1}| = |\overline{h}_{j-1}|$.

	We prove that we can construct a vertex cover for $\bar{G}_j$ by adding at most $\jo(Y_j)$ elements to $\overline{VC}_{j-1}$, which implies $|VC_j| \le |\overline{h}_{j-1}| + \jo(Y_j)$.
	Consider vertex cover instance $\bar{G}_j$ and set $\overline{VC}_{j-1}$.
	By definition, $\overline{VC}_{j-1}$ covers all edges that are part of partial matching $\bar{h}_{j-1}$.
	
	Consider the elements of $\bar{V}_j$ that are an endpoint of an edge in $\{e,f\} \in \bar{E}_j \sym \bar{E}_{j-1}$ with $e,f$ non-trivial in $G_j$.
	By \Cref{lem:oldedgeerror}, $\oj(e) \ge 1$ for each such $e$ \jnew{and $\jo(Y_j) \ge |U|$ for the set $U$ of all such elements. Thus, we can afford to add $U$ to the vertex cover.}
	
	Next, consider an edge $\{e,f\} \in \bar{E}_j$ that is not covered by $\overline{VC}_{j-1} \cup U$.	
	Since $\{e,f\}$ is not covered by $U$, it must hold that $\{e,f\} \in \bar{E}_j \cap \bar{E}_{j-1}$.
	Thus, $\{e,f\}$ was covered by $VC_{j-1}$ but is not covered by $\overline{VC}_{j-1}$.
	This implies $\{e,f\} \cap VC_{j-1} \not= \emptyset$ but $\{e,f\} \cap \overline{VC}_{j-1} = \emptyset$.
	Assume w.l.o.g.~that $e \in VC_{j-1}$.
	Then, there must be an $e'$ such that $\{e,e'\} \in h_{j-1}$ but $\{e,e'\} \not\in \bar{h}_{j-1}$.
	It follows that $\{e,e'\} \not \in \bar{E}_j$.
	As $\{e,f\}$ is not covered by $U$, the endpoint $e'$ must be trivial in $G_j$ but non-trivial in $G_{j-1}$.
	Thus, $e'$ must have been queried (i) as a mandatory element (or a matching partner) in Line~\ref{line_mst_two_error_man1} or~\ref{line_mst_two_error_man2}, (ii) as part of $VC_{j-1}$ in Line~\ref{line_mst_two_error_query_vc} or (iii) in Line~\ref{line_mst_two_error_rob1}.
	Case (ii) implies $e' \in VC_{j-1}$, contradicting $e \in VC_{j-1}$.
	Cases (i) or (iii) imply a query to the matching partner $e$ of $e'$, which contradicts $\{e,f\} \in \bar{E}_j$ (as $e$ would be trivial).
	Thus, $\{e,f\}$ is covered by $\overline{VC}_{j-1} \cup U$, which implies that $\overline{VC}_{j-1} \cup U$ is a vertex cover for $\bar{G}_j$.
	\jnew{\Cref{lem:oldedgeerror} implies $|U| \le \jo(Y_j)$. So, $|VC_j| \le |\bar{h}_{j-1}| + \jo(Y_j)$ which concludes~the~proof.}
\end{proof}

\Cref{lem:error-consistency,lem:error_rob} imply~\Cref{thm:predfree-error}.
Combining~\Cref{thm:predfree-error,mst_end_of_phase_one}, we show Theorem~\ref{thm:error-sensitive}. 

\section{Further research directions}

%

\smallskip \noindent Plenty other (optimization) problems seem natural in the context of explorable uncertainty with untrusted predictions. For our problem, it would be nice to close the gap in the robustness.
We expect that our results extend 
\nnew{to all matroids as it does in the classical setting.} 
{While we ask 
for the minimum number of queries to solve a problem {\em exactly}, it is natural to ask for approximate solutions.  The bad news is that for 
the MST problem 
there is no improvement over the robustness guarantee of $2$ possible even when allowing an arbitrarily large approximation of the exact solution~\cite[Section 10]{megow17mst}. However, it remains open whether an improved consistency or an error-dependent competitive ratio are possible.}


\bibliography{queries.bib,ml.bib}

\appendix
\section{Missing proofs and additional material for Section~\ref{sec:overview}}

\paragraph*{Lower bound for error measure $k_{\#} = 1$}
\label{app:prelim}
In this appendix we prove the following lemma, which shows
that even if a single prediction is wrong ($k_{\#} = 1$), the competitive ratio cannot be better than the known lower bound of~$2$.
Hence, it is impossible to achieve a smoothly degrading performance ratio with respect
to the error measure $k_{\#}$.

\begin{figure}[b]
  \centering
  \begin{subfigure}[t]{0.3\textwidth}
  \centering
  \scalebox{0.3333}{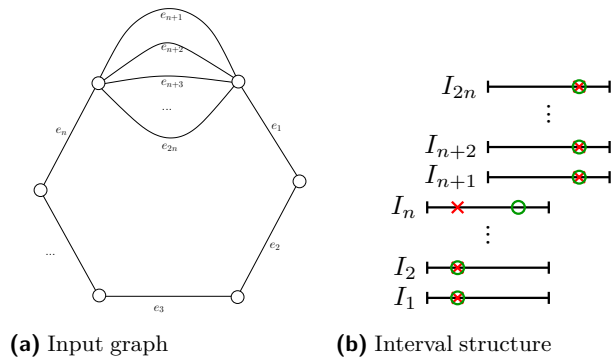}
  \caption{Input graph}
  \label{fig_lb_wrong_predictions_graph}
  \end{subfigure}
  \begin{subfigure}[t]{0.3\textwidth}
  \centering
  \begin{tikzpicture}[line width = 0.3mm, scale=0.8]
    \intervalpr{$I_1$}{0}{2}{0}{0.5}{0.5}
    \intervalpr{$I_2$}{0}{2}{0.5}{0.5}{0.5}
    \intervalpr{$I_n$}{0}{2}{1.5}{0.5}{1.5}
    \path (1, 0.5) -- (1, 1.6) node[font=\normalsize, midway, sloped]{$\dots$};

    \intervalpr{$I_{n+1}$}{1}{3}{2}{2.5}{2.5}
    \intervalpr{$I_{n+2}$}{1}{3}{2.5}{2.5}{2.5}
    \intervalpr{$I_{2n}$}{1}{3}{3.5}{2.5}{2.5}
    \path (2, 2.5) -- (2, 3.6) node[font=\normalsize, midway, sloped]{$\dots$};
  \end{tikzpicture}
  \caption{Interval structure}
  \label{fig_lb_wrong_predictions}
  \end{subfigure}
   \caption{Instance for a lower bound based on the number of inaccurate predictions.
  (\subref{fig_lb_wrong_predictions_graph})~Input graph for MST with uncertainty.
  (\subref{fig_lb_wrong_predictions})~Uncertainty intervals of the $2n$ edges.
  Red crosses indicate predicted values, and green circles show true values.
  \label{fig_lowerbounds}
  }
\end{figure}

\begin{lemma}\label{lem:lb_wrong_predictions}
  Even if $k_{\#} = 1$, then any deterministic algorithm for the MST problem under uncertainty has competitive ratio $\rho\geq 2$.
\end{lemma}

\begin{proof}
Consider an input graph that consists of a path $P$ with $n$ edges $e_1,e_2,\ldots,e_n$ and
$n$ parallel edges $e_{n+1},e_{n+2},\ldots,e_{2n}$ between the endpoints of the path,
see Fig.~\ref{fig_lb_wrong_predictions_graph}. Each edge $e_i$ is given with
an uncertainty interval $I_i$ and a predicted value $\w_{e_i}$. The structure of the
intervals and their predicted values is shown in Fig.~\ref{fig_lb_wrong_predictions}.
If all the predictions are correct, the MST consists of the path $P$, and this
can be verified either by querying the $n$ edges of $P$, or by querying the
$n$ parallel edges between the endpoints of~$P$.
Assume w.l.o.g.\ that an algorithm queries the left-side intervals in the order $I_1, I_2, \ldots, I_n$ and the right-side intervals in the order $I_{n+1}, I_{n+2}, \ldots, I_{2n}$.
Before the algorithm queries~$I_n$ or~$I_{2n}$, the adversary sets all predictions as correct, so the algorithm will eventually have to query~$I_n$ or~$I_{2n}$.
If the algorithm queries~$I_n$ before~$I_{2n}$, then the adversary chooses a value for $I_n$ that forces a query in $I_{n+1}, \ldots, I_{2n}$ (illustrated in Fig.~\ref{fig_lb_wrong_predictions}), and the {predicted} values for the remaining right-side intervals as correct, so the optimum solution only queries $I_{n+1}, \ldots, I_{2n}$.
A symmetric argument holds if the algorithm queries~$I_{2n}$ before~$I_n$.
In either case, the MST is $P$ and the optimal query set consists of $n$ queries, while the algorithm is forced to make $2n$ queries.
Furthermore, a single prediction is incorrect, so $k_{\#}=1$.
\end{proof}

\paragraph*{Learnability of predictions}
\label{app:learning}
In this section, we argue about the learnability of our predictions with regard to the error measure $k_h$. 
%
\nnew{We show the following result.
\begin{theorem}
\label{thm:learning}
{The predicted values $\pred{w}$ for a given instance of MST under explorable uncertainty are PAC learnable w.r.t.~$k_h$ with polynomial time and sample complexity.}
\end{theorem}
}

\nnew{There is given an instance of the MST problem under uncertainty with graph $G=(V,E)$ and a set of uncertainty intervals $\mathcal{I}$ for the edge weights.
}
We assume that the realization $w$ (here, $w$ is a vector of all true values $w_e$ with $e\in E$) of true values for $\mathcal{I}$ is i.i.d.~drawn from an unknown distribution $\ud$, and 
that we can i.i.d.~sample realizations from $\ud$ 
to obtain a training set.
Let $\hs$ denote the set of all possible 
prediction vectors $\pred{w}$, 
with $\w_e \in I_e$ for each $I_e \in \mathcal{I}$.
Let $k_h(w,\pred{w})$ denote the hop distance of the prediction $\pred{w}$ for the realization with the real values $w$.
Since $w$ is drawn from $\ud$, the value $k_h(w,\pred{w})$ is a random variable. 
Our goal is to learn predictions $\pred{w}$ that (approximately) minimize the expected error $\EX_{w \sim \ud}[k_h(w,\pred{w})]$.
As a main result of this section, we show the following theorem, i.e., that the predictions are PAC learnable with regard to $k_h$.
The theorem directly implies \Cref{thm:learning}.

\begin{restatable}{theorem}{TheoLearningHop}
	\label{theo_learnability_hop}
	For any~$\eps, \delta \in (0,1)$, there exists a learning algorithm that, using a training set of size~$m$, returns predictions $\pred{w} \in \hs$, such that
		$\EX_{w \sim \ud}[k_h(w,\pred{w})] \le \EX_{w \sim \ud}[k_h(w,\pred{w}^*)] + \eps$ holds with probability at least~$(1-\delta)$, where~$\pred{w}^* = \arg\min_{\pred{w}' \in \hs} \EX_{w \sim \ud}[k(w,\pred{w}')]$.
		The sample complexity is $m \in \mathcal{O}\left(\frac{(\log(|E|) - \log(\delta/|E|))\cdot |E|^2}{(\eps/|E|)^2}\right)$ and the running time is polynomial in $m$ and $|E|$.
\end{restatable}

Since each $I_e$ is an open interval, there are infinitely many predictions $\pred{w}$, and, thus, the set $\hs$ is also infinite.
In order to reduce the size of $\hs$, we discretize each $I_e$ by fixing a finite number of potentially predicted values $\pred{w}_e$ of $I_e$.
We define the set $\hs_e$  of 
predicted values for $I_e$ as follows.
Let $\{B_1,\ldots,B_l\}$ be the set of lower and upper limits of intervals in $\mathcal{I} \setminus I_e$ that are contained in $I_e$.
Assume that $B_1,\ldots,B_l$ are indexed by increasing value. 
Let $B_0 = L_e$ and $B_{l+1} = U_e$ and, for each $j \in \{0,\ldots,l\}$, let $h_j$ be an arbitrary value of $(B_j,B_{j+1})$.
We define $\hs_i = \{B_1,\ldots,B_l,h_0,\ldots,h_l\}$. 
Since two values $\pred{w}_e,\pred{w}_e' \in (B_j,B_{j+1})$ always lead to the same hop distance for interval $I_e$, there will always be an element of $\hs_e$ that minimizes the expected hop distance for $I_e$.
As $k_h(w,\pred{w})$ is just the sum of the hop distances over all $I_e$, and the hop distances of two intervals $I_e$ and $I_{e'}$ with $e \not= e'$ are independent, restricting $\hs$ to the set $\hs_1 \times \hs_2 \times \ldots \times \hs_{|E|}$ (assuming the edges are enumerated from $1$ to $|E|$) does not affect the accuracy of our predictions.
Each $\hs_e$ contains at most $\mathcal{O}(|E|)$ values, and, thus, the discretization reduces the size of $\hs$ to at most $\mathcal{O}(|E|^{|E|})$.
In particular, $\hs$ is now finite.

To efficiently learn predictions that satisfy Theorem~\ref{theo_learnability_hop}, we again exploit that the hop distances of two intervals $I_e$ and $I_{e'}$ with $e \not= e'$ are independent. This is, because the hop distance of $I_e$ only depends on the predicted value $\pred{w}_e$ and the true value $w_e$, but is independent of all $\pred{w}_{e'}$ and $w_{e'}$ with $e\not= e'$.
Let $\jo_e(w_e,\pred{w}_e)$ denote the hop distance $\jo(e)$ of interval $I_e$ for the predicted value $\pred{w}_e$ and the true value $w_e$, and,
for each $e \in E$, let $\pred{w}^*_e$ denote the predicted value that minimizes $\EX_{w\sim \ud}[\jo_e(w_e,\pred{w}_e)]$.
Since the hop distances of the single intervals are independent, the vector $\pred{w}^*$ then minimizes the expected hop distance of the complete instance.
Thus, if we can approximate the individual $\pred{w}^*_e$, then we can show Theorem~\ref{theo_learnability_hop}.

\begin{lemma}
	\label{lemma_learnability_hop}
	For any~$\eps, \delta \in (0,1)$, and any $e \in E$, there exists a learning algorithm that, using a training set of size~$m  \in \mathcal{O}\left(\frac{(\log(|E|) - \log(\delta))\cdot |E|^2}{\eps^2}\right),$ returns a predicted value $\pred{w}_e \in \hs_e$ in time polynomial in~$|E|$ and~$m$, such that
	$\EX_{w \sim \ud}[\jo_e(w_e,\pred{w}_e)] \le \EX_{w \sim \ud}[\jo_e(w_e,\pred{w}_e^*)] + \eps$ holds with probability at least~$(1-\delta)$, where~$\pred{w}_e^* = \arg\min_{\pred{w}_e \in \hs} \EX_{w \sim \ud}[\jo_e(w_e,\pred{w}_e)]$.
\end{lemma}

\begin{proof}
	We show that the basic \emph{empirical risk minimization (ERM)} algorithm already satisfies the lemma. ERM first i.i.d.~samples a trainingset~$S=\{w^1,\ldots,w^m\}$ of~$m$ true value vectors from~$\ud$.
	Then, it returns the~$\pred{w}_e \in \hs_e$ that minimizes the \emph{empirical error}~$h_S(\pred{w}_e) = \frac{1}{m} \sum_{j=1}^{m} h_e(w^j_e,\pred{w}_e)$.
	
	Recall that, as a consequence of the discretization,~$\hs_e$ contains at most $\mathcal{O}(|E|)$ values.
	Since~$\hs_e$ is finite, and the error function $\jo_e$ is bounded by the interval $[0,|E|]$, it satisfies the \emph{uniform convergence property}; cf.~\cite{Shalev2014}. (This follows also from the fact that~$\hs_e$ is finite and, thus, has finite VC-dimension; cf.~\cite{Vapnik1992}.) 
	This implies that, for~$$m = \left\lceil \frac{2\log(2|\hs_i|/\delta)|E|^2}{\eps^2} \right\rceil \in \mathcal{O}\left(\frac{(\log(|E|) - \log(\delta))\cdot |E|^2}{\eps^2}\right),$$ it holds~$\EX_{w \sim \ud}[\jo_e(w_e,\pred{w}_e)] \le \EX_{w \sim \ud}[\jo_e(w_e,\pred{w}_e^*)] + \eps$ with probability at least~$(1-\delta)$, where~$\pred{w}_e$ is the predicted value learned by ERM (cf.~\cite{Shalev2014,vapnik1999}). 
	As $|\hs_e| \in \mathcal{O}(|E|)$, ERM also satisfies the running time requirements of the lemma.
\end{proof}

\begin{proof}[Proof of Theorem~\ref{theo_learnability_hop}]
	Let $\eps' = \frac{\eps}{|E|}$ and $\delta' = \frac{\delta}{|E|}$.
	Furthermore, let $\hs_{\max} = \arg\max_{\hs_e} |\hs_e|$.
	To learn predictions that satisfy the theorem, we first sample a training set $S=\{w^1,\ldots,w^m\}$ with $m = \left\lceil \frac{2\log(2|\hs_{\max}|/\delta')|E|^2}{\eps'^2} \right\rceil$.
	Next, we apply Lemma~\ref{lemma_learnability_hop} to each $\hs_e$ to learn a predicted value $\pred{w}_e$ that satisfies the guarantees of the lemma for $\eps',\delta'$.
	In each application of the lemma, we use the \emph{same} training set~$S$ that was previously sampled.
	
	For each $\pred{w}_e$ learned by applying the lemma, the probability that the guarantee of the lemma is \emph{not} satisfied is less than $\delta'$.
	By the union bound this implies that the probability that at least one $\pred{w}_e$ with $e \in E$ does not satisfy the guarantee is upper bounded by $\sum_{e \in E} \delta' \le |E| \cdot \delta' = \delta$. 
	Thus, with probability at least $(1-\delta)$, all $\pred{w}_e$ satisfy $\EX_{w \sim \ud}[\jo_e(w_e,\pred{w}_e)] \le \EX_{w \sim \ud}[\jo_e(w_e,\pred{w}^*_e)] + \eps'$.
	Since by linearity of expectations $\EX_{w\sim \ud}[k_h(w,\pred{w}^*)] = \sum_{e \in E} \EX_{w \sim \ud}[\jo_e(w_e,w^*_e)]$, we can conclude that the following inequality, where $\pred{w}$ is the vector of the learned predicted values, holds with probability at least $(1-\delta)$, which implies the theorem:
	\begin{align*}
	\EX_{w\sim \ud}[k_h(w,\pred{w})] &= \sum_{e \in E} \EX_{w \sim \ud}[\jo_e(w_e,\pred{w}_e)]\\
	&\le \sum_{e \in E} \EX_{w \sim \ud}[\jo_e(w_e,\pred{w}^*_e)] + \eps'\\
	&\le \left(\sum_{e \in E} \EX_{w \sim \ud}[\jo_e(w_e,\pred{w}^*_e)]\right) + |E| \cdot \eps'\\
	&\le \EX_{w\sim \ud}[k_h(w,\pred{w}^*)] + \eps.
	\end{align*}
\end{proof}

\section{Missing proofs of~\Cref{sec:mst:prelim}}
\label{app:mst:prelim}
\paragraph*{Witness sets and mandatory queries}

\LemMSTPreprocessing*

\begin{proof}
	Let $T_L$ be a lower limit tree for a given instance $G$ and let $T_U$ be an upper limit tree.
	According to~\cite{megow17mst}, all elements of $T_L \setminus T_U$ are mandatory and we can repeatedly query them for (the adapting) $T_L$ and $T_U$ until $T_L = T_U$. 
	We refer to this process as the \emph{first preprocessing step}.
	
	Consider an $f \in E \setminus T_U$ and the cycle $C$ in $T_U \cup \{f\}$.
	If $f$ is trivial, then the true value $w_f$ is maximal in $C$ and we may delete $f$ without loss of generality.
	Assume otherwise. 
	If the upper limit of $f$ is uniquely maximal in $C$, then $f$ is not part of any upper limit tree.
	If there is an $l \in C$ with $U_f = U_l$, then $T_U' = T_U \setminus \{l\} \cup \{f\}$ is also an upper limit tree.
	Since $T_L \setminus T_U' = \{l\}$, we may execute the first preprocessing step for $T_L$ and $T_U'$.
	We repeatedly do this until each $f \in E \setminus T_U$ is uniquely maximal in the cycle $C$ in $T_U \cup \{f\}$.
	Then, $T_U$ is unique.
	
	To achieve uniqueness for $T_L$, consider some $l \in T_L$ and the cut $X$ of $G$ between the two connected components of $T_L \setminus \{l\}$.
	If $l$ is trivial, then the true value $w_l$ is minimal in $X$ and we may contract 
	$l$ without loss of generality.
	Assume otherwise. 
	If $L_l$ is uniquely minimal in $X$, then $l$ is part of every lower limit tree.
	If there is an $f \in X$ with $L_l = L_f$, then $T_L' = T_L \setminus \{l\} \cup \{f\}$ is also a lower limit tree.
	Since $T_L' \setminus T_U = \{f\}$ follows from $T_L = T_U$, we may execute the first preprocessing step for $T_L'$ and $T_U$.
	We repeatedly do this until \tomc{$L_l$ for} each $l \in T_L$ is uniquely minimal in the cut $X$ of $G$ between the two components of $T_L \setminus \{l\}$.
	Then, $T_L$ is unique.
\end{proof}

\mstTreeChange*

\begin{proof}
			Let $e \in T_L \setminus T_L'$, then $e \in T_L$ and $T_L$ being unique imply that $e$ has the unique minimal lower limit in the cut $X_e$ of $G$ between the two connected components of $T_L \setminus \{e\}$.
			Thus, $e$ is part of any lower limit tree for $G$.
			For $e$ not to be part of $T'_L$, it cannot have the unique minimal lower limit in the cut $X_e$ of $G'$.
			Since querying elements in $X_e \setminus \{e\}$ only increases their lower limits, this can only happen if $e \in Q$.
			
			Let $e \in T_L' \setminus T_L$. Then $T_L = T_U$ and $T_L' = T_U'$ imply $e \in T_U' \setminus T_U$.
			Since $e \not\in T_U$ and $T_U$ is unique, it follows that $e$ has the unique largest upper limit in the cycle $C_e$ of $T_U \cup \{e\}$.
			Thus, $e$ is not part of any upper limit tree for $G$.
			For $e$ to be part of $T_U'$, is cannot have the unique largest upper limit in the cycle $C_e$ of $G'$.
			Since querying elements in $C_e \setminus \{e\}$ only decreases their upper limits, this can only happen if $e \in Q$.
\end{proof}

\lemHopDist*

\begin{proof}
	Consider an instance $G=(V,E)$ with
	uncertainty intervals $I_e=(L_e,U_e)$, true values $w_e$ and predicted
	values $\w_e$ for all $e\in E$.
	Let $E_P$ and $E_M$ be the mandatory queries with respect to
	the predicted and true values, respectively.
	We 
	claim that, for
	every interval $I_e$ of an edge $e \in E_P \sym E_M$ , there is an interval $I_g$
	of an edge $g$ that lies on a cycle with~$e$ such that at least one of the following inequalities holds $w_g \le L_e < \w_g$, $w_g < U_e \le \w_g$, $\w_g \le L_e < w_g$ or $\w_g < U_e \le w_g$. 
	This then implies $\oj(e) \ge 1$.
	
	We continue by proving the claim. Consider an edge $e\in E_P\setminus E_M$. (The argumentation
	for edges in $E_M\setminus E_P$ is symmetric, with the roles of $w$ and $\w$ exchanged.)
	As $e$ is not in $E_M$, replacing all intervals $I_g$
	for $g\in E\setminus\{e\}$ by their true values yields an instance
	that is solved.
	This means that for edge $e$ one of the following
	cases applies:
	\begin{itemize}
		\item[(a)] $e$ is known to be in the MST. Then there is a cut $X_e$
		containing edge $e$ (namely, the cut between the two vertex sets
		obtained from the MST by removing the edge~$e$) such that $e$ is known to
		be a minimum weight edge in the cut, i.e., every other edge $g$ in the cut
		satisfies $w_g\ge U_e$.
		\item[(b)] $e$ is known not to be in the MST. Then there is a cycle $C_e$
		in $G$ (namely, the cycle that is closed when $e$ is added to the MST)
		such that $e$ is a maximum weight edge in $C_e$, i.e., every other
		edge $g$ in the cycle satisfies $w_g\le L_e$.
	\end{itemize}
	As $e$ is in $E_P$, replacing all intervals $I_g$
	for $g\in E\setminus\{e\}$ by their predicted values yields an instance~$\Pi$
	that is not solved. Let $T'$ be the minimum spanning tree of
	$G'=(V,E\setminus\{e\})$ for~$\Pi$. Let $C'$ be the cycle
	closed in $T'$ by adding $e$, and let $f$ be an edge with the largest predicted
	value in $C'\setminus\{e\}$. Then there are only
	two possibilities for the minimum spanning tree of $G$ for $\Pi$: Either
	$T'$ is also a minimum spanning tree of $G$ (if $\w_e\ge \w_f$),
	or the minimum spanning tree is $T'\cup\{e\}\setminus\{f\}$.
	As knowing whether $e$ is in the minimum spanning tree would
	allow us to determine which of the two cases applies, it
	must be the case that we cannot determine whether $e$
	is in the minimum spanning tree or not without querying~$e$.
	If $e$ satisfied case (a) with cut $X_e$ above, then there must be an edge
	$g$ in $X_e\setminus\{e\}$ with $\w_g<U_e$, because otherwise
	$e$ would also have to be in the MST of $G$ for $\Pi$, a contradiction.
	Thus, $\w_g < U_e \le w_g$.
	If $e$ satisfied case (b) with cycle $C_e$ above, then there must be an edge
	$g$ in $C_e\setminus\{e\}$ with $\w_g>L_e$, because otherwise
	$e$ would also be excluded from the MST of $G$ for $\Pi$, a contradiction.
	Thus, $w_g \le L_e < \w_g$.
	In conclusion $\oj(e) \ge 1$, which establishes claim and lemma.
\end{proof}  

\paragraph*{Identifying witness sets}

	We introduce new structural properties to identify witness sets.
Existing algorithms for MST under uncertainty~\cite{erlebach08steiner_uncertainty,megow17mst} essentially follow the algorithms of Kruskal or Prim, and only identify witness sets in the cycle or cut that is currently under consideration.
Let $f_1,\ldots,f_l$ denote the edges in $E\setminus T_L$ ordered by non-decreasing lower limit.
	Then, $C_i$ with $i \in \{1,\ldots,l\}$ denotes the unique cycle in $T_L \cup \{f_i\}$.
	Furthermore, define $G_i = (V,E_i)$ with $E_i = T_L \cup \{f_1,\ldots,f_i\}$.
	For each $e \in T_L$, let $X_e$ denote the set of edges in the cut of the two connected components of $T_L \setminus \{e\}$.
	Existing algorithms for MST under explorable uncertainty repeatedly consider (the changing) $C_1$ or $X_e$, where $e$ is the edge in $T_L$ with maximum upper limit, and identify the maximum or minimum edge in the cycle or cut by querying witness sets of size two, until the problem is solved. 
	For our algorithms, we need to identify witness sets in cycles $C_i \not= C_1$ and cuts $X_{e'} \not= X_e$.
		Thus, the following two lemmas \nnew{alone} are not sufficient for our purpose.

\begin{lemma}[{\cite[Lemma~4.1]{megow17mst}}]
	\label{lemma_mst_1}
	Let $i \in \{1,\ldots,l\}$. 
	Given a feasible query set $Q$ for the uncertainty graph $G=(V,E)$, the set $Q_i := Q \cap E_i$ is a feasible query set for $G_i = (V,E_i)$.
\end{lemma}

\begin{lemma}[{\cite[Lemma~4.2 and ff.]{megow17mst}}]
	\label{lemma_mst_2}
	For some realization of edge weights, let $T_i$ be an 
	MST for graph $G_i$ and let $C$ be the cycle closed by adding $f_{i+1}$ to $T_i$.
	Further, let $h$ be some edge with the largest upper limit in $C$ and $g \in C \setminus \{h\}$ be an edge with $U_g > L_h$. 
	Then any feasible query set for $G_{i+1}$ contains $h$ or $g$.
	Moreover, if $I_g$ is contained in $I_h$, then edge $h$ is mandatory. Further, $f_{i+1}$ has the largest upper limit on $C$ after querying $Q_{i}= E_{i} \cap Q$ for a feasible query set $Q$.
\end{lemma}

Lemma~\ref{lemma_mst_2} shows how to identify a witness set on the cycle closed by $f_{i+1}$ after an MST for graph $G_{i}$ has been verified. Known algorithms for MST under uncertainty build on this by iteratively resolving cycles closed by 
edges $f_1,\ldots,f_l$ one after the other (and analogously for cut-based algorithms). We design algorithms that query edges with a less local strategy; this requires to identify witness sets involving edges $f_{i+1}$ without first verifying an MST for~$G_{i}$. The  following two lemmas provide new structural insights that are fundamental for our algorithms.

\begin{restatable}{lemma}{LemMstWitnessSec}
	\label{lemma_mst_witness_set_2}
	Let $l_i \in C_i \setminus \{f_i\}$ with $I_{l_i} \cap I_{f_i} \not= \emptyset$ such that $l_i \not\in C_j$ for all $j < i$, then $\{l_i,f_i\}$ is a witness set.
	Furthermore, if $w_{l_i} \in I_{f_i}$, then $\{f_i\}$ is a witness set.
\end{restatable}

\begin{proof}
	Consider the set of edges $X_{i}$ in the cut of $G$ defined by the two connected components of $T_L \setminus \{l_i\}$. 
	By assumption, $l_i,f_i \in X_i$. However, $f_j \not\in X_i$ for all $j < i$,  
	as otherwise $l_i \in C_j$ for an $f_j \in X_i$ with $j < i$, which contradicts  the assumption.
	We 
	observe $X_i \cap E_i = \{f_i,l_i\}$. 
	
	Let $Q$ be any feasible query set.
	By Lemma~\ref{lemma_mst_1}, $Q_{i-1} = E_{i-1} \cap Q$ verifies an MST $T_{i-1}$ for $G_{i-1}$. 
	Consider the unique cycle $C$ in $T_{i-1} \cup \{f_i\}$. 
	By Lemma~\ref{lemma_mst_2}, $f_i$ has the highest upper limit on $C$ after querying $Q_{i-1}$. 
	Since $f_i \in X_i$, $f_i \in C$ and $C$ is a cycle, it follows that another edge in $X_i\setminus \{f_i\}$ must be part of $C$.
	We already observed, $X_i \cap E_i = \{l_i,f_i\}$, and therefore $l_i \in C$. 
	Lemma~\ref{lemma_mst_2} implies that $\{f_i,l_i\}$ is a witness set. 
	If $w_{l_i} \in I_{f_i}$, then $f_i$ must be queried to identify the maximal edge in $C$, so it follows that $\{f_i\}$ is a witness set.
\end{proof}

The proof of Lemma~\ref{lemma_mst_witness_set_1} for cycles is very similar to the proof of Lemma~\ref{lemma_mst_witness_set_2} for cuts. However, it requires the following additional observation.

\begin{obs}\label{obs_mst_1}
	Let $Q$ be a feasible query set that verifies an MST $T^*$. 
	Consider any path $P \subseteq T_L$ between two endpoints $a$ and $b$, and let $e \in P$ be the edge with the highest upper limit in $P$. 
	If $e \not\in Q$, then the unique path $\hat{P} \subseteq T^*$ from $a$ to $b$ is such that $e \in \hat{P}$ and $e$ has the highest upper limit in $\hat{P}$ after $Q$ has been queried. 
\end{obs}

\begin{proof}
	For each $i \in \{0,\ldots,l\}$ let $T^*_i$ be the MST for $G_i$ as verified by $Q_i = Q \cap E_i$ and let $\hat{C_i}$ be the unique cycle in $T^*_{i-1} \cup \{f_i\}$. 
	Then $T^*_i = T^*_{i-1} \cup \{f_i\} \setminus \{h_i\}$ holds where $h_i$ is the maximal edge on $\hat{C_i}$.
	Assume $e \not\in Q$. 
	We claim that there cannot be any $\hat{C_i}$ with $e \in \hat{C_i}$ such that $Q_i$ verifies that an edge $e' \in \hat{C_i}$ with $U_{e'} \le U_e$ is maximal in $\hat{C_i}$. 
	Assume otherwise.
	If $e'\not= e$, $Q_i$ would need to verify that $w_{e} \le w_{e'}$ holds. 
	Since $U_{e'} \le U_{e}$, this can only be done by querying $e$, which is a contradiction to $e \not\in Q$.
	If $e'=e$, then $\hat{C}_i$ still contains edge $f_i$.
	Since $T_L = T_U$, $f_i$ has a higher lower limit than $e$. 
	To verify that $e$ is maximal in $\hat{C}_i$, $Q_i$ needs to prove $w_e \ge w_{f_i} > L_{f_i}$. 
	This can only be done by querying $e$, which is a contradiction to $e \not\in Q$.
	
	We show via induction on $i \in \{0,\ldots,l\}$ that each $T^*_i$ contains a path $P^*_i$ from $a$ to $b$ with $e \in P^*_i$ such that $e$ has the highest upper limit in $P^*_i$ after $Q_i$ has been queried. 
	For this proof via induction we define $Q_0 = \emptyset$. 
	\emph{Base case $i=0$}: Since $G_0 = (E,T_L)$ is a spanning tree, $T^*_0 = T_L$ follows. 
	Therefore $P^*_0 = P$ is part of $T_L$ and by assumption $e \in P$ has the highest upper limit in $P^*_0$.
	
	\emph{Inductive step}: By induction hypothesis, there is a path $P^*_i$ from $a$ to $b$ in $T^*_i$ with $e \in P^*_i$ such that $e$ has the highest upper limit in $P^*_i$ after querying $Q_i$.
	Consider cycle $\hat{C}_{i+1}$. 
	If an edge $e' \in \hat{C}_{i+1} \setminus P^*_i$ is maximal in $\hat{C}_{i+1}$, then $T^*_{i+1} = T^*_i \cup \{f_{i+1}\} \setminus \{e'\}$ contains path $P^*_i$. 
	Since $e$ by assumption is not queried, $e$ still has the highest upper limit on $P^*_i = P^*_{i+1}$ after querying $Q_{i+1}$ and the statement follows.

	Assume some $e' \in P^*_i \cap \hat{C}_{i+1}$ is maximal in $\hat{C}_{i+1}$, then $U_{e'} \le U_e$ follows by induction hypothesis since $e$ has the highest upper limit in $P^*_i$. 
	We already observed that $\hat{C}_{i+1}$ then cannot contain $e$. 
	Consider $P' = \hat{C}_{i+1} \setminus P^*_i$. 
	Since $e' \in P^*_i$ is maximal in $\hat{C}_{i+1}$, we can observe that $P' \subseteq T^*_{i+1}$ holds.  
	It follows that path $P^*_{i+1} = P' \cup (P^*_i \setminus\hat{C}_{i+1})$ with $e \in P^*_{i+1}$ is part of $T^*_{i+1}$.
	Since $e$ is not queried, it still has a higher upper limit than all edges in $P^*_i$. 
	Additionally, we can observe that after querying $Q_{i+1}$ no $u \in P'$ can have an upper limit $U_{u} \ge U_{e'}$. 
	If such an $u$ would exist, querying $Q_{i+1}$ would not verify that $e'$ is maximal on $\hat{C}_{i+1}$, which contradicts the assumption. 
	Using $U_e \ge U_{e'}$, we can conclude that $e$ has the highest upper limit on $P^*_{i+1}$ and the statement follows.
\end{proof}

\LemMstWitnessFir*

\begin{proof}
	To prove the lemma, we have to show that each feasible query set contains at least one element of $\{f_i,l_i\}$. Let $Q$ be an arbitrary feasible query set. 
	By Lemma~\ref{lemma_mst_1}, $Q_{i-1} := Q \cap E_{i-1}$ is a feasible query set for $G_{i-1}$ and verifies some MST $T_{i-1}$ for $G_{i-1}$. 
	We show that $l_i \not\in Q_{i-1}$ implies either $l_i \in Q $ or $f_i \in Q$. 	
	
	Assume $l_i \not\in Q_{i-1}$ and let $C$ be the unique cycle in $T_{i-1} \cup \{f_i\}$.
	Since $T_L = T_U$, edge $f_i$ has the highest upper limit in $C$ after querying $Q_{i-1}$.
	While we only assume $T_L = T_U$ for the initially given instance, \Cref{lemma_mst_2} implies that $f_i$ still has the highest upper limit in $C$ after querying $Q_{i-1}$. 
	If we show that $l_i\not\in Q_{i-1}$ implies $l_i \in C$, we can apply Lemma~\ref{lemma_mst_2} to derive that $\{f_i,l_i\}$ is a witness set for graph $G_i$, and thus either $f_i \in Q_i \subseteq Q$ or $l_i \in Q_i \subseteq Q$.
	For the remainder of the proof we show that $l_i\not\in Q_{i-1}$ implies $l_i \in C$. 
	Let $a$ and $b$ be the endpoints of $f_i$, then the path $P = C_i \setminus \{f_i\}$ from $a$ to $b$ is part of $T_L$ and $l_i$ has the highest upper limit in $P$. 
	Using $l_i \not\in Q_{i-1}$ we can apply Observation~\ref{obs_mst_1} to conclude that there must be a path $\hat{P}$ from $a$ to $b$ in $T_{i-1}$ such that $l_i$ has the highest upper limit on $\hat{P}$ after querying $Q_{i-1}$. 
	Therefore, $C = \hat{P} \cup \{f_i\}$ and it follows $l_i \in C$.
	
	If $w_{f_i} \in I_{l_i}$ and $l_i \not\in Q_{i-1}$, then $l_i$ must be queried to identify the maximal edge on $C$, thus it follows that $\{l_i\}$ is a witness set.
\end{proof}

\LemMSTPredFreeIff*

\begin{proof}
	We start by showing the first part of the lemma, i.e., that an instance $G$ is prediction mandatory free if and only if 
	$\pred{w}_{f_i} \ge U_{e}$ and $\pred{w}_e \le L_{f_i}$ holds for each $e \in C_i \setminus \{f_i\}$ and each cycle $C_i$ with ${i} \in \{1,\ldots,l\}$.

	For the first direction, assume $\pred{w}_{f_i} \ge U_{e}$ and $\pred{w}_e \le L_{f_i}$ holds for each $e \in C_i \setminus \{f_i\}$ and each cycle $C_i$,  ${i} \in \{1,\ldots,l\}$.
	Then each $f_i \in E \setminus T_L$ is predicted to be maximal on $C_i$ and each $e \in T_L$ is predicted to be minimal in $X_e$.
	Assuming the predictions are correct, we observe that each vertex cover of 
	$\bar{G}$ is a feasible query set~\cite{erlebach14mstverification}, where 
	$\bar{G} = (\bar{V},\bar{E})$ with $\bar{V} = E$ {(excluding trivial edges)} and $\bar{E} = \{ \{f_i,e\} \mid i \in \{1,\ldots,l\}, e \in C_i \setminus \{f_i\} \text{ and } I_e \cap I_{f_i} \not= \emptyset\}$. 
	Since both $Q_1 := T_L$ and $Q_2 := E \setminus T_L$ are vertex covers for $\bar{G}$, $Q_1$ and $Q_2$ are feasible query sets under the assumption that the predictions are correct.
	This implies that no element is part of every feasible solution because $Q_1 \cap Q_2 = \emptyset$.
	We can conclude that no element is prediction mandatory and the instance is prediction mandatory free.
	
	Next, We show that instance $G$ being prediction mandatory free implies that
		$\pred{w}_{f_i} \ge U_{e}$ and $\pred{w}_e \le L_{f_i}$ holds for each $e \in C_i \setminus \{f_i\}$ and each cycle $C_i$ with ${i} \in \{1,\ldots,l\}$; via contraposition.	
	Assume there is a cycle $C_i$ such that $\w_{f_{i}} \in I_e$ or $\w_e \in I_{f_i}$ for some $e \in C_i \setminus \{e\}$.
	Let $C_i$ be such a cycle with the smallest index.
	If $\w_{f_i} \in I_e$ for some $e \in C_i \setminus \{e\}$, then also $\w_{f_i} \in I_{l_i}$ for the edge $l_i$ with the highest upper limit in $C_i \setminus \{f_i\}$.
	(This is because we assume $T_L = T_U$.)
	Under the assumption that the predictions are true, Lemma~\ref{lemma_mst_witness_set_1} implies that $l_i$ is mandatory and thus prediction mandatory.
	It follows that $G$ is not prediction mandatory free.
	
	Assume $\w_e \in I_{f_i}$.
	We can conclude $e \not\in C_j$ for each $j < i$.
	This is because $\w_e \in I_{f_i}$ and $j < i$ would imply $\w_e \in I_{f_j}$.
	As we assumed that $C_i$ is the first cycle with this property, $e \in C_j$ leads to a contradiction.
	Under the assumption that the predictions are true, Lemma~\ref{lemma_mst_witness_set_2} implies that ${f}_i$ is mandatory and thus prediction mandatory.
	It follows that $G$ is not prediction mandatory free.
%
%
	
	We show the second part of the lemma, i.e., that, once an instance is prediction mandatory free, it remains so even if we query further elements, as long as we maintain unique $T_L = T_U$.
	We show this part ´by proving the following claim.
	
	\noindent\textbf{Claim:} \emph{Let $G$ be a prediction mandatory free instance with unique $T_L = T_U$, and let $G'$ be an instance with unique $T'_L = T_U'$ that is obtained from $G$ by querying a set of edges $Q$, where $T_L'$ and $T_U'$ are the lower and upper limit trees of $G'$.
			Then $G'$ is prediction mandatory free.}
	
	Let $G = (V,E)$ and $G'=(V',E')$ as well as $T_L = T_U$ and $T_L' = T_U'$ be as described in the claim.
	We show that $G$ being prediction mandatory free implies that $G'$ is prediction mandatory free via proof by contradiction.
	Therefore, assume that $G'$ is not prediction mandatory free.
	
	Let $T_L'$ be the lower limit tree of $G'$, let $f_1',\ldots,f'_{l'}$ be the (non-trivial) edges in $E' \setminus T_L'$ ordered by lower limit non-decreasingly, and let $C_i'$ be the unique cycle in $T_L' \cup \{f_i'\}$.
	By assumption, $T_L' = T_U'$ holds and $T_L' = T_U'$ is unique.
	We can w.l.o.g.\ ignore trivial edges in $E' \setminus T_L'$ since those are maximal in a cycle and can be deleted. 
	Since $G'$ is not prediction mandatory free, there must be some $C_i'$ such that either $\w_e \in I_{f_i'}$ or $\w_{f_i'} \in I_e$ for some non-trivial $e \in C_i' \setminus \{f_i'\}$ (cf.~\Cref{mst_pred_free_characterization}).
	
	Assume $e \not\in T_L$.
	Since $e$ is part of $T_L' = T_U'$, Lemma~\ref{lemma_mst_tree_change} implies that $e$ must have been queried and therefore is trivial, which is a contradiction.
	Assume $e \in T_L$ and $\w_e \in I_{f_i'}$.
	As $G$ is prediction mandatory free, the assumption implies $e \not\in C_{f_i'}$ and, therefore, $f_{i'} \not\in X_e$, where $X_e$ is the cut between the two components of $T_L \setminus \{e\}$ in $G$.
	Thus (since $e \in T_L$), cycle $C_i'$ must contain some $f \in X_e \setminus \{e\}$ with $f \not= f_i'$.
	Note that $f \in X_e \setminus \{e\}$ implies $e \in C_f$, where $C_f$ is the cycle in $T_L \cup \{f\}$.
	By~\Cref{mst_pred_free_characterization}, instance $G$ being prediction mandatory free implies $\w_e \not\in I_f$ where $I_f$ denotes the uncertainty interval of $f$ before querying it.
	If $\w_e \in I_{f_i'}$, this implies $L_{f_i'} < L_{f}$. It follows that $f$ has the highest lower limit in $C_i'$, which contradicts $f_i'$ having the highest lower limit in $C'_i$.

	Assume $e \in T_L$ and $\w_{f_i'} \in I_e$. 
	Remember that $f_i'$ is non-trivial and $f_i' \not\in T_L' = T_U'$. 
	According to Lemma~\ref{lemma_mst_tree_change}, it follows $f_i' \not\in T_L = T_U$.
	Let $C_{f_i'}$ be the cycle in $T_L \cup \{f_i'\}$.
	Since $G$ is prediction mandatory free, $\w_{f_i'} \not\in I_{e'}$ for each $e' \in C_{f_i'} \setminus \{f_i'\}$, which implies $U_e > U_{e'}$.
	It follows that the highest upper limit on the path between the two endpoints of $f_i'$ in $T_U' = T_L'$ is strictly higher than the highest upper limit on the path between the two endpoints of $f_i'$ in $T_U=T_L$.
	We argue that this cannot happen and we have a contradiction to $e \in T_L$ and $\pred{f'_i} \in I_e$.
	
	Let $P$ be the path between the endpoints $a$ and $b$ of $f_i'$ in $T_U$ and let $P'$ be the path between $a$ and $b$ in $T_U'$.
	Define $U_P$ to be the highest upper limit on $P$.
	Observe that the upper limit of each edge can only decrease from $T_U$ to $T_U'$ since querying edges only decreases their upper limits.
	Therefore each $e' \in P' \cap P$ cannot have a higher upper limit than $U_P$.
	It remains to argue that the upper limit of each $e' \in P' \setminus P$ cannot be larger than $U_P$.
	Consider the set $\mathcal{S}$ of maximal subpaths $S \subseteq P'$ such that $P \cap P' = \emptyset$.
	Each $e' \in P' \setminus P$ is part of such a subpath $S$.
	Let $S$ be an arbitrary element of $\mathcal{S}$, then there is a cycle $C \subseteq S \cup P$ with $S \subseteq C$.
	Assume $e' \in S$ has a strictly larger upper limit than $U_P$, then an element of $S$ has the unique highest upper limit on $C$.
	It follows that subpath $S$ and path $P'$ cannot be part of any upper limit tree in the instance $G'$, which contradicts the assumption of $P'$ being a path in $T_U'$.
	We conclude that the graph $G'$ is prediction mandatory free.
\end{proof}

\section{Making instances prediction mandatory free}
\label{app:mst-predfree}

Algorithm~\ref{ALG_mst_part_1}
transforms arbitrary instances into prediction mandatory free instances with the following guarantee.

	\begin{figure}[th]
			\begin{minipage}{0.24\textwidth}
			\centering
			\begin{tikzpicture}[line width = 0.3mm, scale = 0.8, transform shape]
			\intervalp{$I_{f_i}$}{-3}{-0.5}{0}{-0.7}
			\intervalp{$I_{l_i}$}{-3.5}{-1}{0.5}{-3.3}
			\intervalp{}{-4}{-1.5}{1}{-3.2}
			\path (-2.75, 1.6) -- (-2.75, 1.6) node[font=\LARGE, midway, sloped]{$\dots$};
			\end{tikzpicture}
			\end{minipage}
			\begin{minipage}{0.24\textwidth}
				\centering
				\begin{tikzpicture}[line width = 0.3mm, scale = 0.8, transform shape]
				\intervalp{$I_{f_i}$}{-3}{-0.5}{0}{-2.5}
				\intervalp{$I_{l_i}$}{-3.5}{-1}{0.5}{-2.8}
				\interval{}{-4}{-1.5}{1}
				\path (-2.75, 1.6) -- (-2.75, 1.6) node[font=\LARGE, midway, sloped]{$\dots$};
				\node[] at (-4.75,1.8){$(b)$};
				\end{tikzpicture}
			\end{minipage}
			\begin{minipage}{0.24\textwidth}
				\centering
				\begin{tikzpicture}[line width = 0.3mm, scale = 0.8, transform shape]
				\intervalp{$I_{f_i}$}{-3}{-0.5}{0}{-2.5}
				\intervalp{$I_{l_i}$}{-3.5}{-1}{0.5}{-3.3}
				\interval{}{-4}{-1.5}{1}
				\path (-2.75, 1.6) -- (-2.75, 1.6) node[font=\LARGE, midway, sloped]{$\dots$};
				\node[] at (-4.75,1.8){$(c)$};
				\end{tikzpicture}
			\end{minipage}
			\begin{minipage}{0.24\textwidth}
				\centering
				\begin{tikzpicture}[line width = 0.3mm, scale = 0.8, transform shape]
				\intervalp{$I_{f_i}$}{-3}{-0.5}{0}{-0.7}
				\intervalp{$I_{l_i}$}{-3.5}{-1}{0.5}{-3.3}
				\intervalp{$I_{l_i'}$}{-4}{-1.5}{1}{-1.75}
				\path (-2.75, 1.6) -- (-2.75, 1.6) node[font=\LARGE, midway, sloped]{$\dots$};
				\node[] at (-4.75,1.8){$(d)$};
				\end{tikzpicture}
			\end{minipage}
	\caption{with predictions indicated as red crosses. $(a)$ Intervals in a prediction mandatory free cycle.  
				$(b)$--$(d)$ Intervals in a cycle that is not prediction mandatory free.
	}
	\label{Ex_mst_phase1_cases_app}
\end{figure}
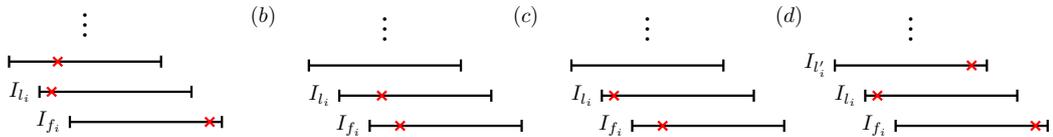

\MSTEndOfPhaseOne*

Note that the Theorem compares $|\ALG|$ with $|(\ALG \cup D) \cap \OPT|$ instead of just $|\OPT|$ since we combine Algorithm~\ref{ALG_mst_part_1} with algorithms for prediction mandatory free instances in the following sections. We will compare the queries of those algorithms against $|\OPT \setminus (\ALG \cup D)|$, which will allow us to analyze the combined algorithms. 

In each iteration the algorithm starts by querying elements that are prediction mandatory for the current instance.
The set of prediction mandatory elements can be computed in polynomial time~\cite{erlebach14mstverification}.
Our algorithm sequentially queries such elements until either $\gamma - 2$ prediction mandatory elements have been queried or no more exist (cf. Line~\ref{line_mst_one_fillup}). 
After each query, the algorithm ensures unique $T_L = T_U$ by using Lemma~\ref{mst_preprocessing}.
Note that the set of prediction mandatory elements with respect to the current instance can change when elements are queried, and therefore we query the elements sequentially.
By \Cref{Theo_hop_distance_mandatory_distance}, each of the {at most} $\gamma-2$ elements is either mandatory or contributes one to the hop distance.  

Afterwards, the algorithm iterates through $i \in \{1,\ldots,l\}$ until the current cycle $C_i$ is not prediction mandatory free. 
If it finds such a cycle~$C_i$, closed by~$f_i$, it queries edges on the cycle and possibly future cycles as follows.
Let $l_i$ denote the edge in $C_i\setminus\{f_i\}$ with highest upper limit. As $C_i$ is not prediction mandatory free, the configuration of $l_i$ and
$f_i$ and their predicted values must be one of those illustrated in Fig.~\ref{Ex_mst_phase1_cases_app}(b)--(c).
In each case, the algorithm queries one to three edges while ensuring $1.5$-consistency
and $2$-robustness \emph{locally} for those queries, as well as a more refined guarantee depending
on prediction errors that will be needed when the algorithm
is used as part of our error-sensitive algorithm in Section~\ref{sec:error-sensitive}.
Line~\ref{line_mst_one_case_a_start} handles the case of Fig.~\ref{Ex_mst_phase1_cases_app}(b),
Lines~\ref{line_mst_one_case_b_start}--\ref{line_mst_one_case_b2} the case of Fig.~\ref{Ex_mst_phase1_cases_app}(c),
and Lines~\ref{line_mst_one_case_c_start}--\ref{line_mst_one_case_c2} the case of Fig.~\ref{Ex_mst_phase1_cases_app}(d).
We now briefly sketch the analysis of these three cases and defer the formal
statements and proof details to Lemmas~\ref{mst_phase1_case_a},
\ref{mst_phase1_case_b} and~\ref{mst_phase1_case_c} in Appendix~\ref{app:mst-predfree}.
In Line~\ref{line_mst_one_case_a}, we can
show that $\{l_i,f_i\}$ is a witness set (giving local $2$-robustness) and either $|\OPT\cap\{l_i,f_i\}|=2$
(giving local $1$-consistency) or $\jo(\{l_i,f_i\})\ge 1$.
For the queries made in Line~\ref{line_mst_one_case_b1}, we can show that if the algorithm
queries three edges, $\OPT$ must query at least two of them (ensuring $1.5$-consistency
and $1.5$-robustness). If the algorithm queries only the two edges $f_i$ and~$l_i$,
they form a witness set (2-robustness) and $\jo(\{f_i,l_i\})\ge 1$.
For Line~\ref{line_mst_one_case_b2}, we show that if the algorithm queries only~$l_i$,
then $\{f_i,l_i\}$ is a witness set and $f_i$ can be deleted from the instance without
querying it (giving $1$-consistency and $1$-robustness). If the algorithm queries
$f_i$ and~$l_i$, they form a witness set and $\jo(\{f_i,l_i\}) \ge 1$.
The guarantees we can prove for queries made in
Lines~\ref{line_mst_one_case_c_start}--\ref{line_mst_one_case_c2} are analogous,
except that the edge $l_i'$ can be contracted instead of deleted if the algorithm
queries only $f_i$ in Line~\ref{line_mst_one_case_c2}.
After processing $C_i$ in this way, the algorithm restarts.
The algorithm terminates when all $C_i$ are prediction mandatory free,
which holds at the latest when all edges in $E$ have been queried. 

\begin{algorithm}[tbh]
	\KwIn{Uncertainty graph $G=(V,E)$ and predictions $\pred{w}_e$ for each {$e \in E$}}
	{Ensure unique $T_L = T_U$. Sequentially query prediction mandatory elements (while ensuring unique $T_L = T_U$) until either $\gamma - 2$ prediction mandatory elements are queried or the instance is prediction mandatory free\label{line_mst_one_fillup}}\;
	Let $T_L$ be the lower limit tree and $f_1,\ldots,f_l$ be the edges in $E \setminus T_L$ ordered by lower limit non-decreasingly\label{line_mst_one_order}\;
	\ForEach{{$C_i$ with $i=1$ to $l$\label{line_mst_one_foreach_f}\label{line_mst_ine_unique_cycle}}}{
		\If{$C_i$ is not prediction mandatory free\label{line_mst_one_if_pred_free}}{
			Let $l_i$ be an edge with highest upper limit in $C_i\setminus\{f_i\}$\;
			\lIf(\tcp*[h]{cf.\ Fig.~\ref{Ex_mst_phase1_cases_app}(b)}){\label{line_mst_one_case_a_start}$\pred{w}_{f_i} \in I_{l_i}$ and $\pred{w}_{l_i} \in I_{f_i}$}{
				Query $\{f_i,l_i\}$\label{line_mst_one_case_a}
			}
			\ElseIf(\tcp*[h]{cf.\ Fig.~\ref{Ex_mst_phase1_cases_app}(c)}){\label{line_mst_one_case_b_start}$\pred{w}_{f_i} \in I_{l_i}$}{
				\If{$\exists l_i' \in C_i \setminus \{f_i,l_i\}$ with $I_{l_i'} \cap I_{f_i} \not= \emptyset$}{	
					$l_i' \gets$ edge in  $C_i \setminus \{f_i,l_i\}$ with the largest upper limit\;	
					Query $\{f_i,l_i\}$, query $l_i'$ only if $w_{f_i} \in I_{l_i}$ and $w_{l_i} \not\in I_{f_j}$ for all $j$ with $l_i \in C_j$\;\label{line_mst_one_case_b1}
				}
				\lElse{
					Query $l_i$, query $f_i$ only if $w_{l_i} \in I_{f_i}$\label{line_mst_one_case_b2}
				}
				
			}
			\ElseIf(\tcp*[h]{cf.\ Fig.~\ref{Ex_mst_phase1_cases_app}(d)}){\label{line_mst_one_case_c_start}$\pred{w}_{l'_i} \in I_{f_i}$ for some $l_i' \in C_i$}{
				$l_i' \gets$ edge with the largest upper limit in $\{l \in C_i\setminus \{f_i\} \mid \pred{w}_{l} \in I_{f_i} \}$\;
				\If{$\exists f_j \in X_{l_i'} \setminus \{f_i,l'_i\}$ with $I_{f_j} \cap I_{l'_i} \not= \emptyset$}{
					$f_j \gets$ edge in $X_{l_i'} \setminus \{f_i,l'_i\}$ with the smallest lower limit\;
					Query $\{f_i,l'_i\}$, query $f_j$ only if $w_{l_i'} \in I_{f_j}$ and $w_{f_i} \not\in I_{e}$ for all $e \in C_i$\;\label{line_mst_one_case_c1}
				}
				\lElse{
					Query $f_i$, query $l'_i$ only if $w_{f_i} \in I_{l'_i}$\label{line_mst_one_case_c2}
				}
			}
			Restart at Line~\ref{line_mst_one_fillup}\;
			
		}
	}
	\caption{Algorithm to make instances prediction mandatory free}
	\label{ALG_mst_part_1}
\end{algorithm}

%
%
%
%

We sketch the proof of Theorem~\ref{mst_end_of_phase_one}.
Elements queried in Line~\ref{line_mst_one_fillup} to ensure unique $T_L=T_U$ are mandatory by Lemma~\ref{mst_preprocessing} and can be ignored in the analysis.
Each iteration of the algorithm queries a set $P_i$ of up to $\gamma-2$ prediction mandatory
edges $e$ in Line~\ref{line_mst_one_fillup}, each of which is mandatory or satisfies
$\oj(e)\ge 1$ by \Cref{Theo_hop_distance_mandatory_distance}, showing
$|P_i|\le |P\cap \OPT|+\oj(P_i)$. In the last iteration, these are the only
queries, and they contribute the additive term $\gamma-2$ to the bound.
In each iteration prior to the last, a set $W_i$ of at most 3 queries
is made in Line~\ref{line_mst_one_case_a}, \ref{line_mst_one_case_b1}, \ref{line_mst_one_case_b2}, \ref{line_mst_one_case_c1} or~\ref{line_mst_one_case_c2}.
These cases are covered by the following possibilities:
(1) The set $W_i$ consists of three edges, and $\OPT$ contains at least two
of them, giving $|W_i|\le 1.5 \cdot |\OPT\cap W_i|$ and $|W_i|\le |\OPT\cap W_i|+1$;
(2) The set $W_i$ consists of two edges, and either $\OPT$ contains
both or $\OPT$ contains one of them and $\jo(W_i)\ge 1$, giving
$|W_i|\le 2\cdot |\OPT\cap W_i|$ and $|W_i|\le |\OPT\cap W_i|+\jo(W_i)$;
(3) The set $W_i$ contains a single edge $e$ and we can delete or contract
another edge $g(e)$ such that $\{e,g(e)\}$ is a witness set,
giving $|W_i|\le |\OPT\cap \{e,g(e)\}|$. Combining these bounds over all
iterations yields Theorem~\ref{mst_end_of_phase_one}.

To conclude the section, we observe that all edges queried by Algorithm~\ref{ALG_mst_part_1} can, w.l.o.g., be contracted or deleted:
Since all queried edges are trivial and we ensure unique $T_L=T_U$, we can observe that each queried $e \in T_L=T_U$ is minimal on a cut and can be contracted, and each queried $e \not\in T_L=T_U$ is maximal on a cycle and can be deleted.
This allows us to treat the instance after the execution of the algorithm independently of all previous queries.

In the following we present the formal lemmas stating that, in each iteration of Algorithm~\ref{ALG_mst_part_1}, the queries made in one excution of Lines~\ref{line_mst_one_case_a_start}--\ref{line_mst_one_case_c2} locally 
satisfy $1.5$-  consistency and $2$-robustness.
All lemmas consider a cycle $C_i$ such that all $C_j$ with $j < i$ are prediction mandatory free, {$l_i$ is the edge with the highest upper limit in $C_i \setminus \{f_i\}$} and predictions are as indicated
in Figure~\ref{Ex_mst_phase1_cases_app}(b) (Lemma~\ref{mst_phase1_case_a}),
Figure~\ref{Ex_mst_phase1_cases_app}(c) (Lemma~\ref{mst_phase1_case_b}),
and Figure~\ref{Ex_mst_phase1_cases_app}(d) (Lemma~\ref{mst_phase1_case_c}).

\begin{lemma}
        \label{mst_phase1_case_a}
        Let $\{f_i,l_i\}$ denote a pair of edges queried in Line~\ref{line_mst_one_case_a}, then $\{f_i,l_i\}$ is a witness set, and either both edges are mandatory or        $\jo(\{f_i,l_i\}) \ge 1$.
\end{lemma}

\begin{proof}
        By assumption, all $C_j$ with $j < i$ are prediction mandatory free.
        We claim that this implies $l_i \not\in C_j$ for all $j < i$.
        Assume, for the sake of contradiction, that there is a $C_j$ with $j < i$ and $l_i \in C_j$.
        Then, $T_L = T_U$ and \tomc{$j < i$} imply that $f_i$ and $f_j$ have larger upper and lower limits than $l_i$ and, since $L_{f_i} \ge L_{f_j}$, it follows $I_{l_i}    \cap I_{f_i} \subseteq I_{l_i} \cap I_{f_j}$.
        Thus, $\pred{w}_{l_i} \in I_{f_i}$ implies $\pred{w}_{l_i} \in I_{{f}_j}$, which contradicts $C_j$ being prediction mandatory free.
        According to Lemma~\ref{lemma_mst_witness_set_2}, $\{f_i,l_i\}$ is a witness set.

        Consider any feasible query set $Q$, then $Q_{i-1}$ verifies the MST $T_{i-1}$ for graph $G_{i-1}$ and $Q$ needs to identify the maximal edge on the unique cycle $C$  in $T_{i-1} \cup \{f_i\}$.
        Following the argumentation of Lemma~\ref{lemma_mst_witness_set_2}, we can observe $l_i,f_i \in C$.
        Since we assume $T_L = T_U$, we can also observe that $f_i$ has the highest upper limit in $C$.
        By Observation~\ref{obs_mst_1}, $l_i$ has the highest upper limit in $C \setminus \{f_i\}$ after querying $Q_{i-1} \setminus \{l_i\}$.

        If $w_{l_i} \in I_{f_i}$, then $f_i$ is part of any feasible query set according to Lemma~\ref{lemma_mst_witness_set_2}.
        Otherwise, $w_{l_i} \le L_{f_i} < \pred{w}_{l_i}$ and $\jo(l_i) \ge 1$.
        If $w_{f_i} \in I_{l_i}$, then $l_i$ is part of any feasible query set according to Lemma~\ref{lemma_mst_witness_set_1}.
        Otherwise, $w_{f_i} {\ge} U_{l_i} {>} \pred{w}_{f_i}$ and $\jo(f_i) \ge 1$.
        In conclusion, either $\{f_i,l_i\} \subseteq Q$ for any feasible query set $Q$ or $\jo(\{f_i,l_i\}) \ge 1$.
\end{proof}

\begin{restatable}{lemma}{MstPhaseOneCaseB}
        \label{mst_phase1_case_b}
        Let $\{f_i,l_i,l_i'\}$ denote the elements of Line~\ref{line_mst_one_case_b1}.
        If the algorithm queries all three elements, then $|\{f_i,l_i,l_i'\} \cap \OPT| \ge 2$.
        Otherwise,  $\jo(\{f_i,l_i\}) > 0$ and $|\{f_i,l_i\} \cap \OPT| \ge 1$.

        Let $\{l_i,f_i\}$ denote the elements of Line~\ref{line_mst_one_case_b2}.
        If the algorithm queries only $l_i$, then $|\{f_i,l_i\} \cap \OPT| \ge 1$ and $f_i$ can be deleted from the instance without querying it.
        Otherwise, $\jo(\{f_i,l_i\}) > 0$ and $|\{f_i,l_i\} \cap \OPT| \ge 1$.
\end{restatable}

\begin{proof}
        By assumption, all $C_j$ with $j < i$ are prediction mandatory free.
        According to Lemma~\ref{lemma_mst_witness_set_1}, $\{f_i,l_i\}$ is a witness set.

        Consider the first part of the lemma, i.e., the elements $\{f_i,l_i,l_i'\}$ of Line~\ref{line_mst_one_case_b1}.
        Assume first that the algorithm queries all three elements.
        By Line~\ref{line_mst_one_case_b1}, this means that $w_{f_i} \in I_{l_i}$ and $w_{l_i} \not\in I_{f_j}$ for each~$j$ with $l_i \in C_j$.
        According to Lemma~\ref{lemma_mst_witness_set_1}, $w_{f_i} \in I_{l_i}$ implies that $l_i$ is part of any feasible query set.
        Consider the relaxed instance where $l_i$ is already queried, then $w_{l_i} \not\in I_{f_j}$ for each $j$ with $l_i \in C_j$ implies that $l_i$ is minimal in          $X_{l_i}$ and that the lower limit tree does not change by querying $l_i$.
        It follows that $l_i'$ is the edge with the highest upper limit in $C_i \setminus \{f_i\}$ in the relaxed instance and, by Lemma~\ref{lemma_mst_witness_set_1}, $\{f_i,l_i'\}$ is a witness set.
        Thus, $|\{f_i,l_i,l_i'\} \cap \OPT| \ge 2$.

        Next, assume that the algorithm queries only $\{f_i,l_i\}$.
        Then, either $w_{f_i} \not\in I_{l_i}$ or $w_{l_i} \in I_{f_j}$ for some $j$ with $l_i \in C_j$.
        Note that if $w_{l_i} \in I_{f_j}$ holds for some $j$ with $l_i \in C_j$, then, by the ordering of the edges in $T_L\setminus E$, the statement holds for some $j \le i$.
        If $w_{f_i} \not\in I_{l_i}$, then $w_{f_i} \ge U_{l_i} > \pred{w}_{f_i}$ and $\jo(f_i) \ge 1$ follows.
        If $w_{l_i} \in I_{f_j}$, then $\pred{w}_{l_i} \le L_{f_j} < w_{l_i}$ and $\jo(l_i) \ge 1$ follows. 
        Therefore, if either  $w_{f_i} \not\in I_{l_i}$ or $w_{l_i} \in I_{f_j}$ for some $j$ with $l_i \in C_j$, then $\jo(f_i,l_i) \ge 1$.

        Consider the second part of the lemma, i.e., the elements $\{f_i,l_i\}$ of Line~\ref{line_mst_one_case_b2}.
        Assume first that the algorithm queries only $l_i$.
        Clearly, $|\{f_i,l_i\} \cap \OPT| \ge 1$ as $\{f_i,l_i\}$ is a witness set.
        Consider the cycle $C_i$.
        As $f_i$ is not queried, it follows $w_{l_i} \not\in I_{f_i}$.
        Furthermore, by definition of Line~\ref{line_mst_one_case_b2}, it holds $I_{l_i'} \cap I_{f_i} = \emptyset$ for all $l_i' \in C_i \setminus \{f_i,l_i\}$.
        Thus, $f_i$ is proven to be uniquely maximal on cycle $C_i$ and, therefore, can be deleted from the instance.
        Next, assume that the algorithm queries $l_i$ and $f_i$.
        Clearly, $|\{f_i,l_i\} \cap \OPT| \ge 1$ still holds.
        As $f_i$ is queried, we have $w_{l_i} > L_{f_i} \ge \pred{w}_{l_i}$ and, therefore, $\jo(\{f_i,l_i\}) \ge 1$.
\end{proof}

Using an analogous strategy, we can show the following lemma.

\begin{restatable}{lemma}{MstPhaseOneCaseC}
        \label{mst_phase1_case_c}
        Let $\{f_i,l_i',f_j\}$ denote the elements of Line~\ref{line_mst_one_case_c1}.
        If the algorithm queries all three elements, then $|\{f_i,l_i',f_j\} \cap \OPT| \ge 2$.
        Otherwise,  $\jo(\{f_i,l'_i\}) > 0$ and $|\{f_i,l'_i\} \cap \OPT| \ge 1$.

        Let $\{l'_i,f_i\}$ denote the elements of Line~\ref{line_mst_one_case_c2}.
        If the algorithm queries only $f_i$, then $|\{f_i,l'_i\} \cap \OPT| \ge 1$ and $l'_i$ can be contracted from the instance without querying it.
        Otherwise, $\jo(\tomc{\{f_i,l'_i\}}) > 0$ and $|\tomc{\{f_i,l'_i\}} \cap \OPT| \ge 1$.
\end{restatable}

\begin{proof}
		By assumption, all $C_j$ with $j < i$ are prediction mandatory free. 
		We claim that this implies $l_i' \not\in C_j$ for all $j < i$. 
		Assume, for the sake of contradiction, that there is a $C_j$ with $j < i$ and $l_i' \in C_j$. 
		Then, $T_L = T_U$ and \tomc{$j < i$} imply that $f_i$ and $f_j$ have larger upper and lower limits than $l_i'$ and, since $L_{f_i} \ge L_{f_j}$, it follows $I_{l_i'} \cap I_{f_i} \subseteq I_{l_i'} \cap I_{f_j}$.
		Thus, $\pred{w}_{l_i'} \in I_{f_i}$ implies $\pred{w}_{l_i'} \in I_{{f}_j}$, which contradicts $C_j$ being prediction mandatory free.
		According to Lemma~\ref{lemma_mst_witness_set_2}, $\{f_i,l_i'\}$ is a witness set.

		Consider the first part of the lemma, i.e., the elements $\{f_i,l_i',f_j\}$ of Line~\ref{line_mst_one_case_c1}.
		Assume first that the algorithm queries all three elements.
		By Line~\ref{line_mst_one_case_c1}, this means that $w_{l_i'} \in I_{f_i}$ and $w_{f_i} \not\in I_{e}$ for each $e \in C_i$. 
		According to Lemma~\ref{lemma_mst_witness_set_2}, $w_{l_i'} \in I_{f_i}$ implies that $f_i$ is part of any feasible query set. 
		Consider the relaxed instance where $f_i$ is already queried, then $w_{f_i} \not\in I_{e}$ for each $e \in C_i$ implies that $f_i$ is maximal in $C_i$ and that the lower limit tree does not change by querying $f_i$. 
		It follows that $f_j$ is the edge with the smallest index (observe that $j>i$ must hold by the argument at the beginning of the proof) and $l_i' \in C_j$ in the relaxed instance and, by Lemma~\ref{lemma_mst_witness_set_2}, $\{f_j,l_i'\}$ is a witness set. 
		Thus, $|\{f_i,l'_i,f_j\} \cap \OPT| \ge 2$.
		Next, assume that the algorithm queries only $\{f_i,l'_i\}$.
		Then, either $w_{l_i'} \not\in I_{f_i}$ or $w_{f_i} \in I_{e}$ for some $e \in  C_i$.
		If $w_{l_i'} \not\in I_{f_i}$, then $w_{l_i'} \le L_{f_i} < \w_{l_i'}$ and $\jo(l_i') \ge 1$ follows 
		If $w_{f_i} \in I_{e}$, then $\w_{f_i} \ge U_e > w_{f_i}$ and $\jo(f_i) \ge 1$ follows.
		Therefore  $w_{l_i'} \not\in I_{f_i}$ or $w_{f_i} \in I_{e}$ for some $e \in C_i$ implies $\jo(\{f_i,l_i'\}) \ge 1$.
	
		Consider the second part of the lemma, i.e., the elements $\{f_i,l'_i\}$ of Line~\ref{line_mst_one_case_c2}.
		Assume first that the algorithm queries only $f_i$.
		Clearly, $|\{f_i,l'_i\} \cap \OPT| \ge 1$ as $\{f_i,l'_i\}$ is a witness set.
		Consider the cut $X_{l_i'}$. 
		As $l'_i$ is not queried, it follows $w_{f_i} \not\in I_{l_i'}$.
		Furthermore, by definition of Line~\ref{line_mst_one_case_c2}, it holds $I_{f_j} \cap I_{l_i'} = \emptyset$ for all $f_j \in X_{l_i'} \setminus \{f_i,l'_i\}$.
		Thus, $l'_i$ is proven to be uniquely minimal in cut $X_{l_i'}$ and, therefore, can be contracted from the instance.
		Next, assume that the algorithm queries $f_i$ and $l'_i$.
		Clearly, $|\{f_i,l'_i\} \cap \OPT| \ge 1$ still holds.
		As $l'_i$ is queried, we have $w_{f_i} < U_{l'_i} \ge \pred{w}_{f_i}$ and, therefore, $\jo(\{f_i,l'_i\}) \ge 1$.
\end{proof}

Using these three lemmas, we give a full proof of Theorem~\ref{mst_end_of_phase_one}.

\begin{proof}[Proof of Theorem~\ref{mst_end_of_phase_one}]
	Since Algorithm~\ref{ALG_mst_part_1} only terminates if Line~\ref{line_mst_one_if_pred_free} determines each $C_i$ to be prediction mandatory free, the instance after executing the algorithm is prediction mandatory free by definition {and Lemma~\ref{mst_pred_free_characterization}}.
	All elements queried in Line~\ref{line_mst_one_fillup} to ensure unique $T_L=T_U$ are mandatory by Lemma~\ref{mst_preprocessing} and never worsen the performance guarantee.
	
	Since the last iteration does not query in Lines~\ref{line_mst_one_case_a}, \ref{line_mst_one_case_b1}, \ref{line_mst_one_case_b2}, \ref{line_mst_one_case_c1} and \ref{line_mst_one_case_c2}, the last iteration queries at most $\gamma-2$ intervals.
	In the following, we consider iterations $i$, that are not the last iteration.
	Each such iteration $i$ queries a set $P_i$ of $\gamma-2$ prediction mandatory elements in Line~\ref{line_mst_one_fillup} and a set $W_i$ in Line~\ref{line_mst_one_case_a}, \ref{line_mst_one_case_b1}, \ref{line_mst_one_case_b2}, \ref{line_mst_one_case_c1} or \ref{line_mst_one_case_c2}.
	By~\Cref{Theo_hop_distance_mandatory_distance}, each $e \in P_i$ is either mandatory or $\oj(e) \ge 1$.

	Consider an arbitrary iteration $i$ (that is not the last one).
	Then, $|P_i| \le |\OPT \cap P_i| + \oj(P_i)$ by the argument above.
	Assume $W_i$ was queried in Line~\ref{line_mst_one_case_a},~\ref{line_mst_one_case_b1} or~\ref{line_mst_one_case_c1}, then \Cref{mst_phase1_case_a,mst_phase1_case_b,mst_phase1_case_c} imply $|W_i| \le \frac{3}{2} \cdot \left(|W_i \cap \OPT| + \jo(W_i)\right)$.
	Thus, $|W_i \cup P_i| \le (1 + \frac{1}{\gamma}) \cdot (|\OPT \cap (W_i \cup P_i)| + \jo(W_i) + \oj(P_i))$.
	At the same time, the lemmas imply that $W_i$ is a witness set, and, therefore, $|W_i \cup P_i| \le \gamma \cdot |\OPT \cap (W_i \cup P_i)|$.
	Putting it together, we get $|W_i \cup P_i| \le \min\{(1 + \frac{1}{\gamma}) \cdot (|\OPT \cap (W_i \cup P_i)| + \jo(W_i) + \oj(P_i)),\gamma \cdot |\OPT \cap (W_i \cup P_i)|\}$.
	Let $\mathcal{K}_1$ denote the union of all sets $W_i \cup P_i$ that satisfy the assumption and let $J_1$ denote the index set of the corresponding sets, then 
	\begin{align*}
	|\mathcal{K}_1| &= \sum_{i \in J_1} |W_i\cup P_i|\\ 
	&\le \sum_{i\in J_1} \min\{(1 + \frac{1}{\gamma}) \cdot (|\OPT \cap (W_i \cup P_i)| + \jo(W_i) + \oj(P_i)),\gamma \cdot |\OPT \cap (W_i \cup P_i)|\}\\
	&\le \min\{(1 + \frac{1}{\gamma}) \cdot (|\OPT \cap \mathcal{K}_1| + \jo(\mathcal{K}_1) + \oj(\mathcal{K}_1)),\gamma \cdot |\OPT \cap \mathcal{K}_1|\}.
	\end{align*}
	
	Assume now that $W_i$ was queried in Line~\ref{line_mst_one_case_b2} or~\ref{line_mst_one_case_c2}.
	Let $e_i \in W_i$ denote the element that is queried first, and let $g(e_i)$ denote the element that either is queried second or deleted/contracted after the first query.
	Since each $g(e_i)$ is either queried or deleted/contracted, no such edge is considered more than once.
	\Cref{mst_phase1_case_b,mst_phase1_case_c} imply $|W_i| \le \left(|(W_i \cup \{g(e_i)\}) \cap \OPT| + \jo(W_i)\right)$ and $|W_i| \le 2 \cdot |(W_i \cup \{g(e_i)\}) \cap \OPT|$.
	Thus, $|W_i \cup P_i| \le (1 + \frac{1}{\gamma}) \cdot (|{\OPT} \cap (W_i \cup \{g(e_i)\} \cup P_i)| + \jo(W_i) + \oj(P_i))$ and $|W_i| \le \gamma \cdot  |(W_i \cup \{g(e_i)\} \cup P_i) \cap \OPT|$.
	Let $\mathcal{K}_2$ denote the union of all sets $W_i\cup P_i$ that satisfy the assumption, let $J_2$ denote the index set of the corresponding sets, and let $\mathcal{G}$ denote the union of all corresponding $\{g(e_i)\}$:
	\begin{align*}
	|\mathcal{K}_2| &= \sum_{W_i \in J_2} |W_i\cup P_i|\\ 
	&\le \sum_{W_i \in J_2} \min\{|\OPT \cap (W_i \cup \{g(e_i)\} \cup P_i)| + \jo(W_i\cup P_i), \gamma \cdot |\OPT \cap (W_i \cup \{g(e_i)\} \cup P_i)|\}\\
	&\le \min\{|\OPT \cap (\mathcal{K}_2 \cup \mathcal{G})|+ \jo(\mathcal{K}_2 ) + \oj(\mathcal{K}_2),\gamma \cdot |\OPT \cap (\mathcal{K}_2 \cup \mathcal{G})|\}.
	\end{align*}
	Let $\mathcal{K}_3$ denote the queries of the last iteration.
	Recall that $D$ is the set of edges in $E\setminus \ALG$ that can be deleted/contracted by the algorithm \emph{before the final iteration.
	By definition, $\mathcal{G} \subseteq D$.}
	Furthermore, note that $|\mathcal{K}_3| \le \gamma -2$ and $|\mathcal{K}_3| \le |\OPT \cap \mathcal{K}_3| + \oj(\mathcal{K}_3)$.
	Since $\mathcal{K}_1$, $\mathcal{K}_2 \cup \mathcal{G}$ and $\mathcal{K}_3$ are pairwise disjoint, we can conclude
	\begin{align*}
	|\ALG_1| &= \mathcal{K}_1 + \mathcal{K}_2 + \mathcal{K}_3\\
	&\le \min\{(1 + \frac{1}{\gamma}) \cdot (|{(\ALG \cup D) \cap} \OPT| + \jo(\ALG) + \oj(\ALG))\\
	&\phantom{\le\min\{},\gamma \cdot |{(\ALG \cup D) \cap} \OPT| + \gamma-2\}.
	\end{align*}
\end{proof}

\section{Missing proofs and examples of~\Cref{sec:optimal-tradeoff}}
\label{app:mst-optimal-tradeoff}
\ThmLBTradeoffWithoutError*

	\begin{proof}
		Assume, for the sake of contradiction, that there is a deterministic $\beta$-robust algorithm that is $\alpha$-consistent with $\alpha=1 + \frac{1}{\beta}-\eps$, for some $\eps>0$.
		Consider the instance that consists of a single cycle with the edges $e_0,e_1,\ldots,e_\beta$ and uncertainty intervals as indicated in Figure~\ref{fig_combined_lb_min_single_set}. 

		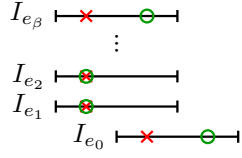
\begin{figure}
			\centering
			\begin{tikzpicture}[line width = 0.3mm, scale=0.8]
					\intervalpr{$I_{e_1}$}{1}{3}{0}{1.5}{1.5}
					\intervalpr{$I_{e_2}$}{1}{3}{0.5}{1.5}{1.5}
					\intervalpr{$I_{e_\beta}$}{1}{3}{1.5}{1.5}{2.5}
					\path (2, 0.5) -- (2, 1.6) node[font=\normalsize, midway, sloped]{$\dots$};
					\intervalpr{$I_{e_0}$}{2}{4}{-0.5}{2.5}{3.5}
			\end{tikzpicture}		
			\caption{Uncertainty intervals for the lower bound example of \Cref{theo_minimum_combined_lb}. Red crosses indicate predicted values, and green circles show true values.}
			\label{fig_combined_lb_min_single_set}
		\end{figure}

		The algorithm must query the edges $\{e_{1}, \ldots, e_{\beta}\}$ first as otherwise, it would query $\beta+1$ intervals in case all predictions are correct, while there is an optimal query set of size $\beta$, which would imply a consistency of $1 + \frac{1}{\beta} > \alpha$.  
		
		Suppose w.l.o.g.\ that the algorithm queries the edges $\{e_{1}, \ldots, e_{\beta}\}$ in order of increasing indices. Consider the adversarial choice $w_{e_i} = \pred{w}_{e_i}$, for $i = 1, \ldots, \beta - 1$, and then $w_{e_\beta} \in I_{e_0}$ and $w_{e_0} \notin I_{e_1} \cup \ldots \cup I_{e_\beta}$.
		This forces the algorithm to query also~$I_0$, while an optimal solution only queries~$I_0$.
		Thus any such algorithm has robustness at least~$\beta+1$, a contradiction.
		
		The second part of the theorem directly follows from the first part and the known general lower bound of $2$ on the competitive ratio~\cite{erlebach08steiner_uncertainty,kahan91queries}. Assume there is an $\alpha$-consistent deterministic algorithm with some $\alpha=1 + \frac{1}{\beta'}$, for some $\beta'\in [1,\infty)$. Consider the instance above with $\beta=\beta'-1$. Then the algorithm has to query edges $\{e_{1}, \ldots, e_{\beta}\}$ first to ensure $\alpha$-consistency as otherwise it would have a competitive ratio of $\frac{\beta+1}{\beta}>1+\frac{1}{\beta'}=\alpha$ in case that all predictions are correct. By the argumentation above, the robustness factor of the algorithm is at least $\beta+1=\beta'=\frac{1}{\alpha-1}$.
	\end{proof}

\MstPhaseTwTwo*

\begin{proof}[Proof of the second statement]
	Consider an arbitrary $f'_i$ and $h(f'_i)$ with $i \le b$.
	By definition of $h$, the edge $h(f'_i)$ is part of the lower limit tree. 
	{Let $X_i$ be the cut between the two components of $T_L \setminus \{h(f'_i)\}$, then we claim that $X_i$ only contains $h(f'_i)$ and edges in $\{f'_{1},\ldots,f'_{g}\}$ (and possibly irrelevant edges that do not intersect $I_{h(f'_i)}$).
		To see this, assume an $f_j \in \{f_1,\ldots,f_l\}$ with $f_j \not\in VC$ was part of $X_i$. 
		Since $f_j \not\in VC$, each edge in $C_j \setminus \{f_j\}$ must be part of $VC$ as otherwise $VC$ would not be a vertex cover. 
		But if $f_j$ is in cut $X_i$, then $C_j$ must contain $h(f'_i)$. 
		By definition, $h(f'_i) \not\in VC$ holds  which is a contradiction. 
		We can conclude that $X_i$ only contains $h(f'_i)$ and edges in $\{f'_{1},\ldots,f'_{g}\}$.}
	
	Now consider any feasible query set $Q$, then $Q_{i'-1}$ verifies an MST $T_{i'-1}$ for graph $G_{i'-1}$ where $i'$ is the index of $f'_i$ in the order $f_1,\ldots,f_l$. 
	Consider again the cut $X_i$. 
	Since each $\{f'_1,\ldots,f'_{i-1}\}$ is maximal in a cycle by assumption, $h(f'_{i})$ is the only edge in the cut that can be part of $T_{i'-1}$.
	Since one edge in the cut must be part of the MST, we can conclude that $h(f'_i) \in T_{i'-1}$. 
	Finally, consider the cycle $C$ in $T_{i'-1}\cup \{f'_i\}$. 
	We can use that $f'_i$ and $h(f'_i)$ are the only elements in $X_i \cap (T_{i'-1}\cup \{f'_i\})$ to conclude that $h(f'_i) \in C$ holds. 
	According to Lemma~\ref{lemma_mst_2}, $\{f'_i,h(f'_i)\}$ is a witness set. 
\end{proof}

\optTradeoff*

\begin{proof}
	We show that the algorithm that first executes Algorithm~\ref{ALG_mst_part_1} and then Algorithm~\ref{ALG_mst_part_2} satisfies the theorem.
	Let $\ALG = \ALG_1 \cup \ALG_2$ be the query set queried by the algorithm, where $\ALG_1$ and $\ALG_2$ are the queries of Algorithm~\ref{ALG_mst_part_1} and Algorithm~\ref{ALG_mst_part_2}, respectively. 
	Let $P \subseteq \ALG_1$ denote the edges queried in the last iteration of Algorithm~\ref{ALG_mst_part_1}.
	Furthermore, let $D$ denote the set of edges in $E \setminus \ALG_1$ that can be deleted or contracted during the execution of Algorithm~\ref{ALG_mst_part_1} before the final iteration. 
	It follows $D \cap \ALG_2 = \emptyset$.
	
	Assume first that $\ALG_2 = \emptyset$. 
	Then, querying $\ALG_1$ solves the problem.
	Theorem~\ref{mst_end_of_phase_one} directly implies $(1+\frac{1}{\gamma})$-consistency. 
	However, due to the additive term of $\gamma-2$ in the second term of the minimum, the theorem does not directly imply $\gamma$-robustness.
	Recall that the additive term is caused exactly by the queries to $P$.
	As the algorithm executes queries in the last iteration, it follows that the instance is not solved at the beginning of the iteration (a solved instance is prediction mandatory free and would lead to direct termination).
	Thus, $|\OPT \setminus ((\ALG_1 \cup D)\setminus P)| \ge 1$ and $|P| \le (\gamma-2) \cdot |\OPT \setminus ((\ALG_1 \cup D)\setminus P)|$.
	Ignoring the additive term caused by $P$, Theorem~\ref{mst_end_of_phase_one} implies $|\ALG_1\setminus P|\le \gamma \cdot |\OPT \cap ((\ALG_1 \cup D)\setminus P)|$.
	As $|\ALG| = |\ALG_1\setminus P| + |P|$, adding up the inequalities implies $|\ALG|\le \gamma \cdot |\OPT|$.
	
	Now, assume that $\ALG_2 \not= \emptyset$.
	Let $\OPT = \OPT_1 \cup \OPT_2$ be an optimal query set with $\OPT_1 = \OPT \cap (\ALG_1 \cup D)$ and $\OPT_2 = \OPT \setminus (\ALG_1 \cup D)$. 
	By Theorem~\ref{thm:predfree}, $\ALG_2 \not= \emptyset$ implies $\OPT_2 \not= \emptyset$.
	Thus, $|\OPT_2| \ge 1$.
	Then, Theorems~\ref{thm:predfree} and~\ref{mst_end_of_phase_one} directly imply $|\ALG_1\setminus P| \le \gamma \cdot |\OPT_1\setminus P|$ (since the additive term in Theorem~\ref{mst_end_of_phase_one} is caused by $P$) and $|\ALG_2 \cup P| \le \gamma \cdot |\OPT_2 |$ (since $|\OPT_2| \ge 1$).
	Together, the inequalities imply $\gamma$-robustness.
	
	In terms of consistency, Theorem~\ref{mst_end_of_phase_one} implies $|\ALG_1| \le (1+\frac{1}{\gamma}) \cdot |\OPT_1|$ and Theorem~\ref{thm:optimal-tradeoff} implies $|\ALG_2| = |\OPT_2|$.
	Together, the inequalities imply $(1+\frac{1}{\gamma})$-consistency.
\end{proof}

\begin{example}
	Figure~\ref{fig:vc_ex1} shows a predictions mandatory instance $G=(V,E)$. The  corresponding vertex cover instance $\bar{G}$ is the complete bipartite graph with the partitions $T_L$ and $E\setminus T_L$. 
	Thus, both $T_L$ and $E\setminus T_L$ are minimum vertex covers of size $n/2$.
	However, querying only $l_1$ and $f_1$ (for the true values as depicted in the figure) solves the instance. 
	So, the size of the initial minimum vertex cover is no lower bound on $|\OPT|$.
\end{example}

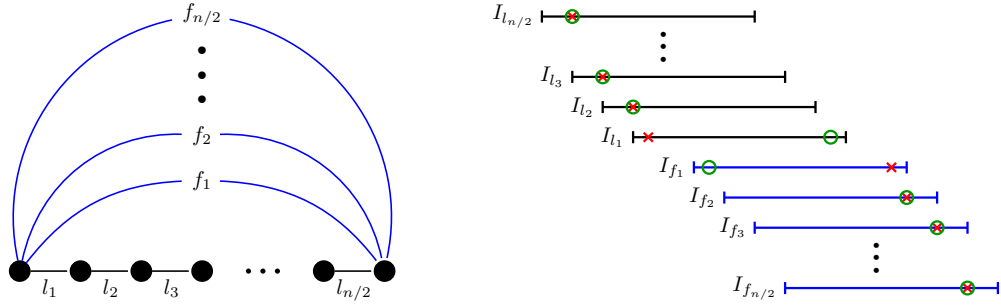
\begin{figure}[h]
	\begin{minipage}{0.35\textwidth}
		\centering
		\begin{tikzpicture}[->,>=stealth',shorten >=1pt,auto,node distance=1cm,
		semithick, 
		condstart/.style = {rectangle,draw,scale=1},
		condend/.style = {condstart},
		stdnode/.style = {circle,draw,fill=black},
		invis/.style = {},, scale = 0.8, transform shape]
		
		\node[stdnode](n1){};
		\node[stdnode,right of = n1](n2){};
		\node[stdnode,right of = n2](n3){};
		\node[stdnode,right of = n3](n4){};
		\node[invis,right of = n4](n5){\Huge$\ldots$};
		\node[stdnode,right of = n5](n6){};
		\node[stdnode,right of = n6](n7){};

		\node[invis,] at (3,1.5)(f1){$f_1$};
		\node[invis,] at (3,2.25)(f2){$f_2$};
		\node[invis,] at (3,3.25)(f3){\small$\vertDOTS$};
		\node[invis,] at (3,4.25)(f3){$f_{n/2}$};
		
		\path
		(n1) edge[-] node[label=below:{$l_1$}]{} (n2)
		(n2) edge[-] node[label=below:{$l_2$}]{} (n3)
		(n3) edge[-] node[label=below:{$l_3$}]{} (n4)
		(n6) edge[-] node[label=below:{$l_{n/2}$}]{} (n7);

		\path
		(n1) edge[-,bend left=25,blue] (f1)
		(f1) edge[-,bend left=25,blue] (n7)	
		
		(n1) edge[-,bend left=40,blue] (f2)
		(f2) edge[-,bend left=40,blue] (n7)	
		
		(n1) edge[-,bend left=45,blue] (f3)
		(f3) edge[-,bend left=45,blue] (n7)	
		;
		
		\end{tikzpicture}		
	\end{minipage}
	\begin{minipage}{0.7\textwidth}
		\centering
		\begin{tikzpicture}
		[line width = 0.3mm, scale = 0.8, transform shape]
		
		\node[] at (-5.5,3)(f3){\tiny$\vertDOTS$};
		\node[] at (-2,-0.5)(f3){\tiny$\vertDOTS$};
		\bintervalpr{$I_{f_1}$}{-5}{-1.5}{1}{-1.75}{-4.75}	
		\bintervalpr{$I_{f_2}$}{-4.5}{-1}{0.5}{-1.5}{-1.5}	
		\bintervalpr{$I_{f_3}$}{-4}{-0.5}{0}{-1}{-1}	
		\bintervalpr{$I_{f_{n/2}}$}{-3.5}{0}{-1}{-0.5}{-0.5}	
		
		\intervalpr{$I_{l_1}$}{-6}{-2.5}{1.5}{-5.75}{-2.75}	
		\intervalpr{$I_{l_2}$}{-6.5}{-3}{2}{-6}{-6}	
		\intervalpr{$I_{l_3}$}{-7}{-3.5}{2.5}{-6.5}{-6.5}	
		\intervalpr{$I_{l_{n/2}}$}{-7.5}{-4}{3.5}{-7}{-7}	
		\end{tikzpicture}
	\end{minipage}
	\caption{Uncertainty graph $G=(V,E)$ (left) and uncertainty intervals of $G$ (right).  
		The black edges ($l_1,\ldots,l_{n/2}$) form the lower limit tree $T_L$ and the blue edges ($f_1,\ldots,f_{n/2}$) are the edges outside of $T_L$.	
		Circles illustrate true values and crosses illustrate the predicted values.}
	\label{fig:vc_ex1}
\end{figure}

\begin{example}
	Figure~\ref{fig:vc_ex2} shows an instance $G=(V,E)$ and the corresponding vertex cover instance $\bar{G}$. 
	With the uncertainty intervals and predictions of Figure~\ref{fig:vc_ex1}, the instance is prediction mandatory free.
	Clearly, the minimum vertex cover of the vertex cover instance is $\{f_1,l_1\}$.
	However, for the true values of Figure~\ref{fig:vc_ex1}, querying $\{f_1,l_1\}$ proves that $f_1$ is part of the MST but $l_1$ is not. 
	After querying $\{f_1,l_1\}$, the new lower limit tree is the path $f_1,l_{n/2},\ldots,l_3,l_2,l_1$ and all edges outside of the lower limit tree connect the endpoint of the path.
	The resulting graph has the same form as the one in Figure~\ref{fig:vc_ex1} and the vertex cover instance is the complete bipartite graph. 
	So, the size of the new minimum vertex cover is $(n/2)-1$, which shows that the size of the minimum vertex cover can increase.	
\end{example}

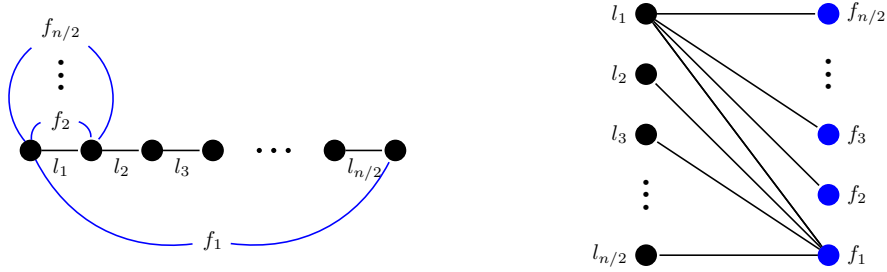
\begin{figure}[h]
	\begin{minipage}{0.5\textwidth}
		\centering
		\begin{tikzpicture}[->,>=stealth',shorten >=1pt,auto,node distance=1cm,
		semithick, 
		condstart/.style = {rectangle,draw,scale=1},
		condend/.style = {condstart},
		stdnode/.style = {circle,draw,fill=black},
		invis/.style = {},scale = 0.8, transform shape]
		
		\node[stdnode](n1){};
		\node[stdnode,right of = n1](n2){};
		\node[stdnode,right of = n2](n3){};
		\node[stdnode,right of = n3](n4){};
		\node[invis,right of = n4](n5){\Huge$\ldots$};
		\node[stdnode,right of = n5](n6){};
		\node[stdnode,right of = n6](n7){};

		\node[invis,] at (3,-1.5)(f1){$f_1$};
		\node[invis,] at (0.5,0.5)(f2){$f_2$};
		\node[invis,] at (0.5,1.25)(f3){\tiny$\vertDOTS$};
		\node[invis,] at (0.5,2)(f3){$f_{n/2}$};
		
		\path
		(n1) edge[-] node[label=below:{$l_1$}]{} (n2)
		(n2) edge[-] node[label=below:{$l_2$}]{} (n3)
		(n3) edge[-] node[label=below:{$l_3$}]{} (n4)
		(n6) edge[-] node[label=below:{$l_{n/2}$}]{} (n7);

		\path
		(n1) edge[-,bend right=35,blue] (f1)
		(f1) edge[-,bend right=35,blue] (n7)	
		
		(n1) edge[-,bend left=40,blue] (f2)
		(f2) edge[-,bend left=40,blue] (n2)	
		
		(n1) edge[-,bend left=45,blue] (f3)
		(f3) edge[-,bend left=45,blue] (n2)	
		;
		
		\end{tikzpicture}		
	\end{minipage}
	\begin{minipage}{0.5\textwidth}
		\centering
		\begin{tikzpicture}[->,>=stealth',shorten >=1pt,auto,node distance=1cm,
				semithick, 
				condstart/.style = {rectangle,draw,scale=1},
				condend/.style = {condstart},
				stdnode/.style = {circle,draw,fill=black},
				invis/.style = {},scale = 0.8, transform shape]

				\node[stdnode,label=left:{$l_1$}](n1){};
				\node[stdnode,below of = n1,label=left:{$l_2$}](n2){};
				\node[stdnode,below of = n2,label=left:{$l_3$}](n3){};
				\node[invis,below of = n3](n4){\tiny$\vertDOTS$};
				\node[stdnode,below of = n4,label=left:{$l_{n/2}$}](n5){};

				\node[stdnode,blue,fill=blue,label=right:{$f_{n/2}$}] at (3,0)(f1){};
				\node[invis,below of = f1](f4){\tiny$\vertDOTS$};
				\node[stdnode,blue,fill=blue,below of = f4,label=right:{$f_3$}](f2){};
				\node[stdnode,blue,fill=blue,below of = f2,label=right:{$f_2$}](f3){};
				\node[stdnode,blue,fill=blue,below of = f3,label=right:{$f_{1}$}](f5){};

				\path
					(n1) edge[-] (f1)
					(n1) edge[-] (f2)
					(n1) edge[-] (f3)
					(n1) edge[-] (f5)

					(f5) edge[-] (n1)
					(f5) edge[-] (n2)
					(f5) edge[-] (n3)
					(f5) edge[-] (n5); 
				
		\end{tikzpicture}
	\end{minipage}
	\caption{Uncertainty graph $G=(V,E)$ (left) and corresponding vertex cover instance $\bar{G}$ (right). The black edges ($l_1,\ldots,l_{n/2}$) form the lower limit tree $T_L$.}
	\label{fig:vc_ex2}
\end{figure}
\section{Missing proofs of~\Cref{sec:error-sensitive}}
\label{app:error-sensitive}
\ThmLBAllErrorMeasures*

\begin{proof}
Consider the instance of Figure~\ref{fig:lb_allmeasures} with the three edges $e_1$, $e_2$ and $e_3$ of a triangle.
Using the depicted predictions, the instance is prediction mandatory free by~\Cref{mst_pred_free_characterization}.
Clearly, edge $e_3$ is part of any MST, so it remains to determine which of $e_1$ and $e_2$ has larger edge weight.
If the algorithm $\ALG$ starts querying~$e_1$, then the adversary picks $w_{e_1} \in I_{e_2}$ and the algorithm is forced to query~$e_2$. Then $w_{e_2} \in I_{e_2} \setminus I_{e_1}$, so the optimum queries only~$e_2$.
It is easy to see that $k_h = 1$ and, thus, $\frac{|\ALG|}{|\OPT|} = 2 = \min\{1+ \frac{k_h}{|\OPT|},2\}$.
A symmetric argument holds if the algorithm starts querying~$e_2$. In that case, $w_{e_2} \in I_{e_1}$ and the algorithm is forced to query~$e_1$. Then $w_{e_1} \in I_{e_1} \setminus I_{e_2}$, so the optimum queries only~$e_1$.
Again, $k_h = 1$ and, thus, $\frac{|\ALG|}{|\OPT|} = 2 = \min\{1+ \frac{k_h}{|\OPT|},2\}$.
Taking multiple copies of this instance (connected by a tree structure), gives the same results for larger values of $k_h$.
\end{proof}

\begin{figure}[h]
	\begin{minipage}{0.45\textwidth}
		\centering
		\begin{tikzpicture}
			[default/.style={draw, fill, circle}, scale=0.9,line width = 0.3mm]
			
			\node[default] (1) at (0,0) {};
			\node[default] (2) at +(2,0) {};
			\pgfmathsetmacro{\yc}{sin(60)}
			\node[default] (3) at +(1,2*\yc) {};	
			\path 
			(1) edge node[anchor=north, below] 		{$e_3$}	(2)
			(2) edge node[anchor=southwest, right] 	{$e_2$}	(3)
			(1)	edge node[anchor=southeast, left]  		{$e_1$}	(3);
		\end{tikzpicture}
	\end{minipage}
	\begin{minipage}{0.45\textwidth}
		\centering
		\begin{tikzpicture}[line width = 0.3mm, scale=0.8]
			\intervalpr{$I_{e_1}$}{0}{2}{0}{0.5}{1.5}
			\intervalpr{$I_{e_2}$}{1}{3}{0.5}{2.5}{2.5}
			\intervalpr{$I_{e_3}$}{-3}{-2}{1}{-2.5}{-2.5}
		\end{tikzpicture}
	\end{minipage}	
	\caption{Lower bound example for the proof of~\Cref{thm_lb_error_measure}. Circles illustrate true values and crosses illustrate the predicted values. }
	\label{fig:lb_allmeasures}
\end{figure}
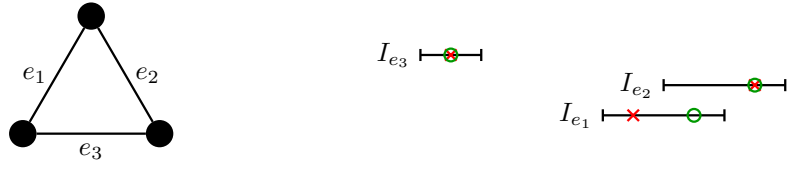

\remEdgeError*

\begin{proof}
	\textbf{Case $\{l,f\} \in \bar{E}_{j-1} \setminus \bar{E}_j$ for non-trivial $l$ and $f$.}
	Assume w.l.o.g.~that $l \in T_L^{j-1}$, $f \not\in T_L^{j-1}$, $l \in C_f^{j-1}$, $f \in X^{j-1}_l$ and $I_l \cap I_f \not= \emptyset$, where $C_f^{j-1}$ is the unique cycle in $T_L^{j-1} \cup \{f\}$ and $X_l^{j-1}$ is the set of edges between the two connected components of $T_L^{j-1}\setminus\{l\}$ 
	.
	By~\Cref{lemma_mst_tree_change}, $l \in T_L^j$ and $f \not\in T_L^j$.
	
	Since $\{l,f\} \not\in \bar{E}_j$, we have $l \not\in C_f^j$ and $f \not\in X_l^j$.
	Each cycle containing $f$ must contain an edge of $X_l^{j-1}\setminus \{f\}$ and, therefore, there must be some $l' \in C_f^j$ with $l' \in X_l^{j-1} \setminus \{f,l\}$. 
	This implies $l' \in T_L^j$.
	On the other hand, $T_L^{j-1} \cap X^{j-1}_l = \{l\}$ and, therefore, $l' \not\in T_L^{j-1}$.
	\Cref{lemma_mst_tree_change} implies that $l'$ must have been queried while transforming instance $G_{j-1}$ into instance $G_j$.
	
	As $G_{j-1}$ is prediction mandatory free, $\w_{l'} \ge U_l > L_f$ by \Cref{mst_pred_free_characterization}.
	However, $w_{l'} \le L_f$ as $w_{l'} > L_f$ is a contradiction to $T_L^j$ being a lower limit tree.
	We can conclude $\w_{l'} \ge U_l > w_{l'}$ and $\w_{l'} > L_f \ge w_{l'}$, which implies $\oj(l),\oj(f) \ge 1$.
	
	\textbf{Case $\{l,f\} \in \bar{E}_j \setminus \bar{E}_{j-1}$.}
	W.l.o.g.~assume $l \in T_L^j$ and $f \not\in T_L^{j}$.
	Since $l$ and $f$ are {non}-trivial, \Cref{lemma_mst_tree_change} implies $e \in T_L^{j-1}$ and $f \not\in T_L^{j-1}$.	

	Since $\{f,l\} \in \bar{E}_j$ and $f \not\in T^j_L$, the definition of the vertex cover instance $\bar{G}_j$ implies $l \in C^j_f$ such that $I_f \cap I_{e} \not= \emptyset$, where $C^j_f$ is the unique cycle in $T^j_L \cup \{f\}$.
		
	Define $X^j_{l}$ to be the set of edges in the cut of $G_j$ between the two connected components of $T_L^j \setminus \{l\}$.
	Remember that $f \in X^j_l$ since $f \not\in T_L^j$ and $l \in C^j_f$.
	Since $\{f,l\} \not\in \bar{E}_{j-1}$ implies $l\not\in C^{j-1}_f$, there must be an element $l' \in T^{j-1}_L \cap (X^j_{l} \setminus \{l,f\})$ such that $l' \in C^{j-1}_f$, where $C^{j-1}_f$ is the unique cycle in $T^{j-1}_L \cup \{f\}$.

	As the instance is prediction mandatory free, $l' \in C^{j-1}_f$ implies $\w_{l'} \not\in I_f$ (more precisely $\w_{l'} \le L_f$).
	Further, $X^j_l \cap T^j_L = \{l\}$ holds by definition and implies $l' \not\in T_L^{j}$.
	According to Lemma~\ref{lemma_mst_tree_change}, $l' \in T^{j-1}_L \setminus T_L^j$ implies that $l'$ must have been queried in order to transform instance $G_{j-1}$ into instance $G_j$.
	If $w_{l'} \in I_l$, then $w_{l'} < U_l$, which implies that $l$ does not have the smallest upper limit in cut $X^{j}_l$. This is a contradiction to $T_L^j$ being an upper limit tree.
	
	If $w_{l'} \not\in I_l$, then also $w_{l'} \ge U_l$ ($w_{l'} \le L_l$ cannot be the case because $l'$ would have the smallest lower limit in $X^j_l$ which would contradict $l' \not\in T_L^j$).
	As $I_l \cap I_f \not= \emptyset$, we can conclude that $w_{l'} \ge U_l$ implies $w_{l'} > L_f$.
	It follows $\w_{l'} \le L_f < w_{l'}$, which implies $\oj(f) \ge 1$.
	
	Similarly, $l \in C_f^j$, $I_l \cap I_f \not= \emptyset$ and $l,f$ being non-trivial imply $L_f < U_l$.
	Thus, we have $\w_{l'} \le L_f < U_l \le w_{l'}$, which implies $\oj(l) \ge 1$.
	
	\jnew{Since all errors of both cases are caused by elements $l$ that have been queried to transform $G_{j-1}$ into $G_j$, i.e., are contained in $Y_j$,
	the arguments in both cases also imply $\jo(Y_j) \ge |U|$ for the set $U$ of all endpoints of edges $\{l,f\}$ that are covered by one of the cases.}	
\end{proof}

\thmErrorSensitive*

We actually show a robustness of $\max\{3,\gamma + \frac{1}{|\OPT|}\}$ which might be smaller than $\gamma +1$.

\begin{proof}
	Let $\ALG = \ALG_1 \cup \ALG_2$ be the query set queried by the algorithm where $\ALG_1$ and $\ALG_2$ are the queries of Algorithm~\ref{ALG_mst_part_1} and Algorithm~\ref{ALG_mst_part_2_error}, respectively. 
	Let $P \subseteq \ALG_1$ denote the edges queried in the last iteration of Algorithm~\ref{ALG_mst_part_1}.
	Furthermore, let $D$ denote the set of edges in $E \setminus \ALG_1$ that can be deleted or contracted during the execution of Algorithm~\ref{ALG_mst_part_1} before the final iteration. 
		It follows $D \cap \ALG_2 = \emptyset$

	Assume first that $\ALG_2 = \emptyset$. 
	Then, querying $\ALG_1$ solves the problem.
	Theorem~\ref{mst_end_of_phase_one} directly implies $|\ALG| = |\ALG_1| \le  (1+\frac{1}{\gamma}) \cdot {(|\OPT {\cap (\ALG_1 \cup D)}| + \jo(\ALG_1) + \oj(\ALG_1))}$ and $|\ALG_1 \setminus P| \le \gamma \cdot |\OPT \cap {(\ALG_1 \cup D)}|$ {(as the additive term of $\gamma-2$ is caused by $P$)}. 
	
	Thus, the error-dependent guarantee follows immediately.
	However, due to the additive term of $\gamma-2$ in the second term of the minimum, the second inequality does not directly transfer to $|\ALG|$ as $|\ALG| \le \gamma \cdot |\OPT|$ might not hold.
	Recall that the additive term is caused exactly by the queries to $P$.
	Since the algorithm executes queries in the last iteration, it follows that the instance is not solved at the beginning of the iteration (a solved instance is prediction mandatory free and would lead to direct termination).
	Thus, $|\OPT \setminus ({(\ALG_1 \cup D)}\setminus P)| \ge 1$ and, therefore $|P| \le (\gamma-2) \cdot |\OPT \setminus {((\ALG_1 \cup D)\setminus P)}|$.
	Ignoring the additive term caused by $P$, Theorem~\ref{mst_end_of_phase_one} implies $|\ALG_1\setminus P|\le \gamma \cdot |\OPT \cap {((\ALG_1 \cup D)\setminus P)}|$.
	We can conclude	
	\begin{align*}
	|\ALG| &= |\ALG_1\setminus P| + |P|\\
	&\le  \gamma \cdot |\OPT \cap {((\ALG_1 \cup D)\setminus P)}| + (\gamma-2) \cdot |\OPT \setminus {((\ALG_1 \cup D)\setminus P)}|\\
	&\le \gamma \cdot |\OPT|.
	\end{align*}
	Now, assume that $\ALG_2 \not= \emptyset$.
	Let $\OPT = \OPT_1 \cup \OPT_2$ be an optimal query set with $\OPT_1 = \OPT \cap {(\ALG_1 \cup D)}$ and $\OPT_2 = \OPT \setminus \ALG_1$. 
	By Theorem~\ref{thm:predfree}, $\ALG_2 \not= \emptyset$ implies $\OPT_2 \not= \emptyset$.
	Thus, $|\OPT_2| \ge 1$.
	Then, Theorems~\ref{mst_end_of_phase_one} and~\ref{thm:predfree-error} directly imply $|\ALG_1 {\setminus P}| \le \gamma \cdot |\OPT_1 \setminus P|$ (since the additive term in Theorem~\ref{mst_end_of_phase_one} is caused by $P$) and $|\ALG_2 \cup P| \le (\gamma+1) \cdot |\OPT_2|$ (since $|\OPT_2| \ge 1$).
	Together, the inequalities imply $|\ALG| \le (\gamma+1) \cdot |\OPT|$ and, thus $(\gamma+1)$-robustness.
	
	We remark that, for $\gamma \ge 3$, the robustness improves for increasing $|\OPT|$.
	To see this, note that if $\gamma \ge 3$, then $|\ALG_2 \cup P|\le \gamma \cdot |\OPT_2| + 1$ and, in combination with $|\ALG_1 \setminus P| \le \gamma \cdot |\OPT_1 \setminus P|$, $|\ALG| \le \gamma \cdot |\OPT| + 1$. 
	Thus, for $\gamma \ge 3$, the robustness of the algorithm is actually $\gamma + \frac{1}{|\OPT|}$. 
	For $\gamma=2$, the robustness is $3 = \gamma + 1$.
	Combined, the robustness is $\max\{3,\gamma + \frac{1}{|\OPT|}\}$.
	
	 {We continue by showing $|\ALG| \le (1+\frac{1}{\gamma}) \cdot  {|\OPT| + \red{5} \cdot k_h}$.
	Theorem~\ref{mst_end_of_phase_one} implies $|\ALG_1| \le (1+\frac{1}{\gamma}) \cdot (|\OPT_1|+ 2 \cdot \max\{\jo(\ALG_1),\oj(\ALG_1)\})$. 
	Since $\gamma \ge 2$, it follows $(1+\frac{1}{\gamma}) \cdot 2 \le 5$, and we can rewrite $|\ALG_1| \le (1+\frac{1}{\gamma}) \cdot |\OPT_1| + \red{5} \cdot \max\{\jo(\ALG_1),\oj(\ALG_1)\}$.
	Theorem~\ref{thm:predfree-error} implies $|\ALG_2| 	\le |\OPT_2|+ \red{5} \cdot k_h'$, where $k_h'$ is the error for the input instance of the second phase. 
	That is, the instance that does not contain any edge of $\ALG_1$ since those can be deleted/contracted after the first phase.
	This implies for the error $k_h$ of the complete instance that no error that is counted by $\max\{\jo(\ALG_1),\oj(\ALG_1)\}$ is considered by $k_h'$ and, therefore, $\max\{\jo(\ALG_1),\oj(\ALG_1)\} + k_h' \le k_h$.
	Thus, we can combine the inequalities to $|\ALG| \le (1+\frac{1}{\gamma}) \cdot |\OPT| + \red{5} \cdot k_h$.}
\end{proof}

\section{Removing the integrality condition from \Cref{thm:optimal-tradeoff,thm:error-sensitive}}
\label{app:rational-gamma}
The parameter $\gamma$ in \Cref{thm:optimal-tradeoff,thm:error-sensitive} is restricted to integral values since it determines sizes of query sets. Nevertheless, a generalization to arbitrary rational $\gamma \geq 2$ is possible by a simple rounding of $\gamma$ at a small loss in the performance guarantee. 

Let $A(\gamma')$ denote the algorithm that first executes Algorithm~\ref{ALG_mst_part_1} and then Algorithm~\ref{ALG_mst_part_2_error}, both with parameter $\gamma=\gamma'$. For given $\gamma \in \ZZ$, we run Algorithm $A(\gamma)$ and achieve the performance guarantee from Theorem \ref{thm:error-sensitive}. Assume $\gamma \notin\ZZ$, and let $\{\gamma\}:=\gamma - \lfloor \gamma \rfloor = \gamma - \lceil \gamma \rceil +1$ denote its fractional part. We randomly chose $\gamma'$ as 
	\[
	\gamma' = \begin{cases}
 	\ \lceil \gamma \rceil & \textup{ w.p. } \{\gamma\}\\
 	\ \lfloor \gamma \rfloor & \textup{ w.p. } 1-\{\gamma\}.
 	\end{cases}
	\]
Then we execute Algorithm $A(\gamma')$. We show that the guarantee from Theorem \ref{thm:error-sensitive} holds in expectation with an additive term less than $0.021$ in the consistency, more precisely, we show the following result generalizing \Cref{thm:error-sensitive}.



\begin{theorem}
	\label{thm:min-problem-arbitrary-gamma}
	For every rational $\gamma\ge 2$, there exists a randomized algorithm $\min\{ 1 + \frac{1}{\gamma} \nnew{+ \xi}+ \frac{\red{5} \cdot k_h}{|\OPT|}, \gamma+1 \}$-competitive algorithm for the MST problem under explorable uncertainty, \nnew{where $\xi=\frac{ \{\gamma\}(1-\{\gamma\}) }{\gamma \lceil \gamma \rceil \lfloor \gamma \rfloor} \leq \frac{1}{48} < 0.021$}.
\end{theorem}

\begin{proof}
	\newcommand{\E}[1]{\mathbb{E}\left[\,#1\,\right]}
	\newcommand{\gfloor}{\lfloor \gamma \rfloor}
	\newcommand{\gceil}{\lceil \gamma \rceil}
	
For a given $\gamma \in \ZZ$, we run Algorithm $A(\gamma)$ and achieve the performance guarantee from Theorem \ref{thm:error-sensitive}. 

Assume $\gamma \notin\ZZ$, and the algorithm $A(\gamma')$ is executed with $\gamma'\in \{\gfloor, \gceil\}$. 
%
For any fixed $\gamma'$, the arguments in the proof of Theorem \ref{thm:error-sensitive} on the robustness imply that the ratio of the algorithm's number of queries $|\ALG|$ and $|\OPT|$ is bounded by $\gamma'+1$. In expectation the robustness is
\begin{align*}
	\E{\gamma'+1} &= (1-\{\gamma\}) \cdot \gfloor + \{\gamma\} \cdot \gceil + 1\\
	&= (1-\{\gamma\}) \cdot (\gamma - \{\gamma\}) + \{\gamma\} \cdot (\gamma - \{\gamma\} +1) + 1\\
	&= \gamma +1.
\end{align*}

Noting that $\opt$ and $k_h$ are independent of~$\gamma'$, the error-dependent bound on the competitive ratio is in expectation 
\begin{align*}
	\E{ 1+\frac{1}{\gamma'} +  \frac{5 \cdot k_h}{\opt} } = 1 + \E{\frac{1}{\gamma'}} +  \frac{5 \cdot k_h}{\opt} .
\end{align*}

Applying simple algebraic transformations, we obtain
\begin{align*}
	\E{\frac{1}{\gamma'}} &= \{\gamma\}\frac{1}{\gceil} + (1-\{\gamma\})\cdot\frac{1}{\gfloor} \ -\  \frac{1}{\gamma} \ +\  \frac{1}{\gamma}\\
	&= \frac{1}{\gamma} + \frac{1}{\gamma \gceil \gfloor}\cdot \bigg( \{\gamma\} \gamma \gfloor  + (1- \{\gamma\}) \cdot \gceil  \gamma - \gceil \gfloor  \bigg).
\end{align*}
As we consider fractional $\gamma >2$, it holds $\lceil \gamma\rceil = \gfloor + 1$. Using this, we rewrite the term in brackets as
\begin{align*}
	& \{\gamma\} \gamma \gfloor  + (1- \{\gamma\}) \cdot (\gfloor +1)  \gamma - (\gfloor +1)  \gfloor\\
	& = \{\gamma\} \gamma \gfloor  + (\gfloor +1)  \gamma - \{\gamma\} \cdot (\gfloor +1)  \gamma - (\gfloor +1)  \gfloor\\
	& = (\gfloor +1)(\gamma -\gfloor) - \{\gamma\} \cdot \gamma\\
	& = \{\gamma\} (\gfloor + 1 - \gamma) \\
	& = \{\gamma\}(1-\{\gamma\}),
\end{align*}
where the third and fourth inequality come from the fact that $\gamma -\gfloor = \{\gamma\}$. Note that for $\{\gamma\} \geq 0$, the expression $\{\gamma\}(1-\{\gamma\})$ is at most $1/4$, where it reaches its maximum for $\{\gamma\} =1/2$. Further, notice that for fractional $\gamma > 2$ it holds that $\gamma \gceil \gfloor \geq 12$. We conclude that
\[
  \E{\frac{1}{\gamma'}} \leq \frac{1}{\gamma} + \frac{ \{\gamma\}(1-\{\gamma\}) }{\gamma \gceil \gfloor} \leq \frac{1}{\gamma} + \frac{1}{48} < \frac{1}{\gamma} + 0.021
\]
which proves the theorem.
\end{proof}

The generalized variant of \Cref{thm:optimal-tradeoff} follows along the same lines. 
\begin{theorem}
	For every rational $\gamma\ge 2$, there exists a $(\nnew{1.021}+\frac{1}{\gamma})$-consistent and $\gamma$-robust algorithm for the MST problem under explorable uncertainty.
\end{theorem}

\end{document}